\documentclass[12pt,twoside]{article}

\usepackage[dvips]{graphicx}
\RequirePackage{color,ragged2e,newlfont,pifont,stmaryrd,xspace}
\RequirePackage{amsmath,amssymb,amsfonts,textcomp,yhmath}
\RequirePackage[mathscr]{eucal} \RequirePackage{eufrak}
\RequirePackage{amsthm}
\RequirePackage{empheq}   
\RequirePackage{mathtools} 
\RequirePackage{cases} 
\RequirePackage[pdftex]{hyperref}\hypersetup{backref,implicit=true,bookmarks=true,
colorlinks=true, linkcolor=blue,filecolor=blue,
pagecolor=blue,urlcolor=blue,pdfpagemode=UseOutlines}

\newtheorem{theorem}{Theorem}[section]
\newtheorem{proposition}[theorem]{Proposition}
\newtheorem{lemma}[theorem]{Lemma}

\newtheorem{remark}[theorem]{Remark}

\makeatletter
\newcounter{systnumber} 
\numberwithin{systnumber}{section} 
\makeatother

\date{}

\begin{document}

\markboth{\small{N. Noutchegueme and  A.
Nangue}}{\small{Einstein-Maxwell-Massive Scalar Field System in 3+1
formulation}}
 \centerline{}


\begin{center}
\Large{{\bf Einstein-Maxwell-Massive Scalar Field System in 3+1
formulation on Bianchi Spacetimes type I-VIII}}
\end{center}

\centerline{}

\centerline{\bf {Norbert Noutchegueme}} \centerline{University of
Yaounde I} \centerline{Department of Mathematics, Faculty of
Science, POB: 812, Yaounde, Cameroon} \centerline{nnoutch@yahoo.fr}
\centerline{}

\centerline{\bf {Alexis Nangue}}
 \centerline{University of Yaounde
I} \centerline{Department of Mathematics, Faculty of Science, POB:
812, Yaounde, Cameroon} \centerline{alexnanga02@yahoo.fr}

\begin{abstract}
Global existence to the coupled Einstein-Maxwell-Massive Scalar
Field system which rules the dynamics of a kind of charged pure
matter in the presence of a massive scalar field is proved, in
Bianchi I-VIII spacetimes; asymptotic behaviour, geodesic
completeness, energy conditions are investigated in the case of a
cosmological constant bounded from below by a strictly negative
constant depending only on the ma-\\ssive scalar field.
\end{abstract}
\noindent
{\bf MR Subject Classification}: 83CXXX 

rential system, charged particles, constraints,
asymptotic behaviour, energy conditions.

\section*{Introduction}\label{intro}
\indent

 In relativistic kinetic theory, global dynamics of
several kinds of charged and uncharged matter remain an active
research domain in General Relativity (GR), in which cosmology plays
one of the central roles, by coupling various matter fields to the
Einstein equations to provide mathematical explanations in response
to new astrophysical observations. In this context, spatially
homogeneous phenomena such as the one we consider in the present
paper are relevant. There are several reasons why it is of interest
to consider the Einstein equations with cosmological constant and to
couple the equations to a massive scalar field.

$\bullet$ Astrophysical observations have made evident the fact
that, even in the presence of material bodies, the gravitational
field can propagate through space at the speed of the light,
analogously to electromagnetic waves. A mathema-\\tical way to model
this phenomenon is to couple a scalar field to the Einstein
equations. Let us recall that the Nobel prize of Physics 1993 was
awarded for works on this subject. More details on this question can
be found in \cite{5}, \cite{21}. In fact, several authors realized
the interest of coupling scalar field to other fields equations; see
for instance \cite{4}, \cite{18}, \cite{20}, \cite{24}, \cite{8}.

$\bullet$ Now our motivation for considering the Einstein equations
with a cosmological constant $\Lambda$ is due to the fact that
astrophysical observations, based on luminosity via redshift plots
of some far away objets such as Supernovae-Ia, have made evident the
fact that the expansion of the universe is accelerating, as foreseen
by E.P. Hubble. A classical mathematical tool to model this
phenomenon is to include the cosmological constant $\Lambda$ in the
Einstein Equations. Several authors did it in the case $\Lambda >
0$; see for instance \cite{14}, \cite{7}, \cite{15}, \cite{22},
\cite{23}, \cite{26}. In the present paper, we prove that, in the
presence of a massive scalar field, this result can be extended, not
only to the case $\Lambda = 0$, but also to the case $\Lambda > -
\alpha ^2$, where $\alpha>0$ is a constant depending only on the
potential of the massive scalar field. This result extends and
completes those of \cite{20}.

 Also recall that the recent Nobel prize of Physics, 2011,
was awarded to three Astrophysicists for their advanced research on
this phenomenon of accelerated expansion of the universe.

In fact, we must point out that, the notion of "dark energy" was
introduced in order to provide a physical explanation to this
phenomenon, but the physical structure of this hypothetical form of
energy which is unknown in the laboratories remains an open question
in modern cosmology; so is the question of "dark matter". Also
notice that the scalar fields are considered to be a mechanism
producing accelerated models, non only in "inflation", which is a
variant of the Big-Bang theory including now a very short period of
very high acceleration, but also in the primordial universe.

In this paper, we consider the 3+1 formulation of the Einstein
equations, which allows to interpret the fields equations as the
time history of the first and second fundamental forms of the
3-hypersurfaces of constant times slices, foliating the space-time.
The background space-time is any Bianchi space-time type I to VIII,
since it is proved in \cite{17} that Bianchi IX such as the
Kantowski-Sachs space-time, develops curvature singularities in a
finite proper time, and this  constitutes a major obstacle to our
goal which is to proved global in time  existence of solutions. We
prove in this paper that if the initial value of the mean curvature
is strictly negative and if $\Lambda > -\alpha^2$, then the coupled
Einstein-Maxwell-Scalar Field System has a global in time-solution.
We investigate the asymptotic behaviour which reveals an exponential
growth of the gravitational potentials, confirming the accelerated
expansion of the universe. We prove the geodesic completeness and we
were able to show that the considered model always satisfies the
weak and the dominant energy conditions; we prove that, if $\Lambda>
\beta^2$ where $\beta> 0$ is a constant depending only on the mean
curvature of the space-time, then the considered model also
satisfies the strong energy condition.

The paper is organized as follows:

\begin{itemize}
    \item In section~\ref{para1}, we introduce the coupled system and we
    give some preli-\\minary results.
    \item In section~\ref{para2}, we study the constraints equations, the mean
    curvature and we introduce the Cauchy problem.
    \item In section ~\ref{para3}, we prove the local and the global existence
    of solutions.
    \item In section ~\ref{para4}, we study the asymptotic behaviour.
    \item In section ~\ref{para5}, we study the geodesic completeness.
    \item In section ~\ref{para6}, we study the energy conditions.
    \item Section ~\ref{para7} is the appendix to which we refer for the
    details of the proofs of some important results.
\end{itemize}
\section{Equations and preliminary results}
\label{para1}
\begin{enumerate}
    \item[$\bullet$] Unless otherwise specified, Greek indices $\alpha$,
    $\beta$, $\lambda$, ..., range from 0 to 3 and Latin indices $i$, $j$, $k$,
    ...,
    from 1 to 3. We adopt the Einstein summation convention $a_{\alpha}b^{\alpha}
    \equiv \sum\limits_{\alpha}a_{\alpha}b^{\alpha}$.\\
We consider a time-oriented space-time $(M, \widetilde{g})$ where M
is a four-dimensional manifold and $\widetilde{g}$ the metric tensor
of lorentzian signature $(-, +, +, +)$. The model adopted for our
study is any Bianchi space-time from type I to VIII, and, following
\cite{17}, \cite{19}, \cite{25}, \cite{27}, who studied the
question, we take M on the form: $M= \mathbb{R}\times G$, where G is
a three dimensional simply connected Lie group. We take
$\widetilde{g}$ on the form:
\begin{equation}\label{1.1}
\widetilde{g}=- dt^2+g_{ij}(t)e^{i}\otimes e^{j}
\end{equation}
where $\{e_{i}\}$ is a left invariant frame on G and $\{e^{i}\}$ is
the dual frame; $(g_{ij})$ is a positive definite Riemannian
3-metric depending only on $t$. Now the vector $n=\partial_{t}$
being orthogonal to $G$, we complete the frame $(e_{i})$ on G, to
obtain a frame $(n, e_{i})$ on M. We have $$\widetilde{g}(n,
e_{i})=\widetilde{g}_{0i}=0,$$ confirming the form (\ref{1.1}) of
$\widetilde{g}$.
    \item[$\bullet$]The Einstein-Maxwell-Massive scalar field system with
    cosmological constant $\Lambda$, which rules the evolution of the
    considered charged pure matter can be written, following
    \cite{9}, and denoting with a tilda $(\widetilde{\;\;})$
    quantities on M:
\begin{numcases}
\strut
 \widetilde{R}_{\alpha \beta}-
\frac{1}{2}\widetilde{R}\widetilde{g}_{\alpha \beta}+ \Lambda
\widetilde{g}_{\alpha \beta}= 8\pi( \widetilde{T}_{\alpha \beta}+\widetilde{\tau}_{\alpha \beta}
+\rho \widetilde{u}_{\alpha}\widetilde{u}_{\beta}) \label{1.2}\\
\widetilde{\nabla}_{\alpha} \widetilde{F}^{\alpha\beta}=4\pi\widetilde{J}^{\beta}\label{1.3}\\
 \widetilde{\nabla}_{\alpha} \widetilde{F}_{\beta\gamma}+ \widetilde{\nabla}_{\beta} \widetilde{F}_{\gamma\alpha}
  +\widetilde{\nabla}_{\gamma} \widetilde{F}_{\alpha\beta}=0
 \label{1.4}
 \end{numcases}
where:
\begin{itemize}
    \item[-] (\ref{1.2}) are the Einstein equations, basic equations
    in GR for the unknown metric tensor
    $\widetilde{g}=(\widetilde{g}_{\alpha\beta})$; $\widetilde{R}_{\alpha \beta}$
    is the Ricci tensor, contracted of the curvature tensor, $\widetilde{R}=
    \widetilde{g}^{\alpha\beta}
    \widetilde{R}_{\alpha\beta}$ is the scalar curvature; $\widetilde{T}_{\alpha \beta}$
    and $\widetilde{\tau}_{\alpha \beta}$ whose expressions are
    given below, are respectively, the tensor associated to a
    massive scalar field $\Phi$ which is an unknown function of the
    time $t$, and $\widetilde{\tau}_{\alpha \beta}$ the Maxwell
    tensor associated to the electromagnetic field $\widetilde{F}$.
    \item[-] (\ref{1.3}) and (\ref{1.4}) are the two sets of
    Maxwell equations, basic equations of Electromagnetism, written in covariant form for the
    electromagnetic field  $\widetilde{F}=(\widetilde{F}^{0i},\widetilde{ F}_{ij})$ which is
    a closed unknown antisymmetric 2-form, depending only on the
    time $t$. $\widetilde{F}^{0i}$ and $\widetilde{ F}_{ij}$ are
    respectively its electric and magnetic parts.
$\widetilde{T}_{\alpha \beta}$ and $\widetilde{\tau}_{\alpha \beta}$
are defined by:
\begin{numcases}
\strut \widetilde{T}_{\alpha
\beta}=\widetilde{\nabla}_{\alpha}\Phi\widetilde{\nabla}_{\beta}\Phi
-\frac{1}{2}\widetilde{g}_{\alpha\beta}\Big[\widetilde{\nabla}^{\lambda}\Phi\widetilde{\nabla}_{\lambda}\Phi
+m^2\Phi^2\Big]  \label{1.5}\\
\widetilde{\tau}_{\alpha
\beta}=-\frac{1}{4}\widetilde{g}_{\alpha\beta}\widetilde{F}^{\lambda\mu}\widetilde{F}_{\lambda\mu}+
\widetilde{F}_{\alpha\lambda}\widetilde{F}^{\;\;\;\lambda}_{\beta}\label{1.6}
 \end{numcases}
where in \eqref{1.5} as in \eqref{1.3} and \eqref{1.4},
$\widetilde{\nabla}$ stands for the covariant derivative, or the
Levi-Civita connection in $\widetilde{g}$, $m>0$ is a given constant
called the \textit{mass} of the scalar field $\Phi$; notice that
$\frac{1}{2}m^2\Phi^2$ represents the potential associated to the
scalar field $\Phi$.
\item[-] In (\ref{1.2}), $\rho
\widetilde{u}_{\alpha}\widetilde{u}_{\beta}$ is the tensor
associated to the considered charged pure matter, with $\rho>0$ an
unknown scalar function of the time $t$, standing for the pure
matter proper density, and $\widetilde{u}=(\widetilde{u}^{\beta})$ a
time-like future pointing unit vector which is an unknown function
of the time $t$, representing the material velocities.
    \item[-] In (\ref{1.3}), $\widetilde{J}=(\widetilde{J}^{\beta})$
    stands for the Maxwell current generated by the charged
    particles of pure matter and defined by:
    \begin{equation}\label{1.7}
\widetilde{J}^{\beta}=e \widetilde{u}^{\beta}
    \end{equation}
in which $e\geq 0$ is an unknown scalar function of the time $t$,
standing for the proper charged density of the charged particles.\\
Notice that the Maxwell equations \eqref{1.4} are just the covariant
notation of the relation $d \widetilde{F}=0$, since
$\widetilde{F}$ is a closed 2-form.\\
Now it is well known, see \cite{12}, that the electromagnetic field
de-\\viates the trajectories of the charged particles which are no
longer the geodesics of the space-time as in the empty case, but the
solutions of the following differential system of current flow:

\begin{equation}\label{1.8}
\widetilde{u}^{\alpha}\widetilde{\nabla}_{\alpha}\widetilde{u}^{\beta}=
4\pi\frac{e}{\rho}\widetilde{F}_{\lambda}^{\;\;\;\beta}\widetilde{u}^{\lambda}
\end{equation}

which reduces to the usual geodesics system when $\widetilde{F}=0$.
\end{itemize}

    \item[$\bullet$]To study the Einstein equations \eqref{1.2},
    we adopt the 3+1 formulation, which allows to interpret these
    equations as the evolution in time of the triplet $(\sum_{t}, g_{t},
    k_{t})$, where $\sum_{t}=\{t\}\times G$, $g_{t}=(g_{ij}(t))$
    stands for the first fundamental form induced on $\sum_{t}$ by
    $\widetilde{g}$, $k_{t}=(k_{ij}(t))$ is the second fundamental
    form defined in the present case by:
\begin{equation}\label{1.9}
    k_{ij}=-\frac{1}{2}\partial_{t}g_{ij}
\end{equation}
where we adopt the ADM/MTV convention see \cite{1},
\cite{6}.\\
By the 3+1 formulation, the Einstein equations \eqref{1.2} which
originally, in the homogeneous case we consider, are a non linear
\textit{second order} differential system in $g_{ij}$, can be
written as an equivalent \textit{first order} differential
system in $(g_{ij}, k_{ij})$ to which standard theory applies.\\
We now introduce a quantity which will play a central role, namely,
the \textit{mean curvature} of the space-time which is a scalar
function denoted by H and defined by:
\begin{equation}\label{1.10}
    H=g^{ij}k_{ij}
\end{equation}
that is the trace of $(k_{ij})$. Next, since
$\widetilde{u}=(\widetilde{u}^{\beta})$ is a unit vector which means
$\widetilde{g}(\widetilde{u},
\widetilde{u})=\widetilde{u}_{\alpha}\widetilde{u}^{\alpha}=-1$, we
deduce from the expression \eqref{1.1} of $\widetilde{g}$ that:
\begin{equation}\label{1.11}
\widetilde{u}^{0}=\sqrt{1+g_{ij}\widetilde{u}^{i}\widetilde{u}^{j}}
\end{equation}
because $\widetilde{u}$ is future pointing, and (\ref{1.11}) shows
that $g_{ij}$ and $u^{i}$ determine $\widetilde{u}^{0}$, and, since
$(g_{ij})$ is positive definite, that $\widetilde{u}^{0}\geq 1$.

    \item[$\bullet$]It is important to note that $(e_{i})$ is not
    always the natural frame $(\partial_{i})$ on G. We set
    $(\widetilde{e}_{\alpha})=(n, e_{i})=(\partial_{t}, e_{i})$:
    that is $\widetilde{e}_{0}=\partial_{t}$ and
    $\widetilde{e}_{i}=e_{i}$. The frames $(\widetilde{e}_{\alpha})$
    and $(\partial_{\alpha})$ are then link by:
\begin{equation}\label{1.12}
\widetilde{e}_{\alpha}=\widetilde{e}^{\lambda}_{\alpha}\partial_{\lambda}
\end{equation}
where
\begin{equation}\label{1.13}
\widetilde{e}^{0}_{0}=1;\;\;\widetilde{e}^{0}_{i}=\widetilde{e}^{i}_{0}=0;\;\;\widetilde{e}^{i}_{j}=e^{i}_{j}.
\end{equation}
Recall that the structure constants $C^{k}_{ij}$ of the Lie algebra
$\mathfrak{g}$ of the Lie group G are defined by:
\begin{equation}\label{1.14}
    [e_{i}, e_{j}]=C^{k}_{ij}e_{k}
\end{equation}
where $[\; ,\; ]$ denotes the Lie brackets of $\mathfrak{g}$. Since
the Lie brackets are antisymmetric, $C^{k}_{ij}$ is antisymmetric
with respect to $i$, $j$ that is:
\begin{equation}\label{1.15}
C^{k}_{ij}=-C^{k}_{ji}.
\end{equation}
Now the Ricci rotation coefficients
$\widetilde{\gamma}^{\lambda}_{\alpha\beta}$ associated to the
Levi-Civita connection $\widetilde{\nabla}$ are defined by:
\begin{equation}\label{1.16}
\widetilde{\nabla}_{\widetilde{e}_{\alpha}}\widetilde{e}_{\beta}=
\widetilde{\gamma}^{\lambda}_{\alpha\beta}\widetilde{e}_{\lambda}.
\end{equation}

We set
\begin{equation}\label{1.17}
 [\widetilde{e}_{\alpha}, \widetilde{e}_{\beta}]=\widetilde{C}^{\lambda}_{\alpha\beta}\widetilde{e}_{\lambda}
\end{equation}
where:
\begin{equation}\label{1.18}
\widetilde{C}^{0}_{\alpha\beta}=\widetilde{C}^{\lambda}_{0\beta}=
\widetilde{C}^{\lambda}_{\alpha0}=0;
\;\;\widetilde{C}^{k}_{ij}=C^{k}_{ij}.
\end{equation}
The general expression of
$\widetilde{\gamma}^{\lambda}_{\alpha\beta}$ in term of
$\widetilde{C}^{\lambda}_{\alpha\beta}$ and
$\widetilde{g}_{\alpha\beta}$ can be find in \cite{2} or in \cite{3}
p.301. It will be enough to extract and mention here only those of
the Ricci rotation coefficients used in the present paper, namely:
\begin{numcases}
\strut
 \widetilde{\gamma}^{\lambda}_{00}=\widetilde{\gamma}^{0}_{i0}=\widetilde{\gamma}^{0}_{0i}=0;\; \;
 \widetilde{\gamma}^{0}_{ij}=-k_{ij} ;\; \;
 \widetilde{\gamma}^{j}_{i0}=\widetilde{\gamma}^{j}_{0i}=-k_{i}^{\;j}\nonumber\\
  \widetilde{\gamma}^{l}_{ij}:=\gamma^{l}_{ij}=\frac{1}{2}g^{lk}\Big[-C^{m}_{jk}g_{im}+
  C^{m}_{ki}g_{jm}+C^{m}_{ij}g_{km}\Big].  \label{1.19}
 \end{numcases}
Notice that the $\widetilde{\gamma}^{\lambda}_{\alpha\beta}$ are not
to be confused with the usual Christoffel symbols
$\Gamma_{\alpha\beta}^{\lambda}$ associated to the natural frame
$(\partial_{\lambda})$. For instance, (\ref{1.19}) gives, using
$C^{k}_{ij}=-C^{k}_{ji}$:
\begin{equation}\label{1.20}
   \gamma^{i}_{ij}=C_{ij}^{i} ; \;\; \gamma^{l}_{ij}-\gamma^{l}_{ji}=C_{ij}^{l}
\end{equation}
which shows that, at the contrary of the
$\Gamma_{\alpha\beta}^{\lambda}$, the
$\widetilde{\gamma}^{\lambda}_{\alpha\beta}$ are not symmetric with
respect to $\alpha$ and $\beta$. We have the following formulae,
analogous to the natural frame case:
\begin{equation}\label{1.21}
\left\{
  \begin{array}{ll}
\widetilde{\nabla}_{\alpha} \widetilde{T}^{\beta}=
\widetilde{e}_{\alpha}(\widetilde{T}^{\beta})
  +\widetilde{\gamma}^{\beta}_{\alpha\lambda}\widetilde{T}^{\lambda} \\
\widetilde{\nabla}_{\alpha}
\widetilde{T}_{\beta}=\widetilde{e}_{\alpha}(\widetilde{T}_{\beta})
  -\widetilde{\gamma}^{\lambda}_{\alpha\beta}\widetilde{T}_{\lambda}.
  \end{array}
\right.
\end{equation}
Also mention the useful formula we prove in Appendix $A_{1}$:
\begin{equation}\label{1.22}
\frac{d}{dt}(\gamma_{ij}^{l})=\nabla^{l}k_{ij}-\nabla_{j}k_{i}^{l}-\nabla_{i}k_{j}^{l}
\end{equation}
where $\nabla$ is the Levi-Civita connection on $(G, g)$.

    \item[$\bullet$]A direct calculation using (\ref{1.21})
    shows that the components of the tensor $\widetilde{T}_{\alpha\beta}$
    defined by (\ref{1.5}) are given, setting
    $\dot{\Phi}=\frac{d\Phi}{dt}$, by:
\begin{equation}\label{1.23}
 \widetilde{T}_{00}=\frac{1}{2}\dot{\Phi}^{2} +
 \frac{1}{2}m^2\Phi^2; \;\; \widetilde{T}_{0i}=0;\;\;
 \widetilde{T}_{ij}=\frac{1}{2}g_{ij}(\dot{\Phi}^{2} -m^2\Phi^2).
\end{equation}
Concerning the Maxwell tensor $\widetilde{\tau}_{\alpha\beta}$
defined by (\ref{1.6}), we first obtain by a direct calculation:

\begin{equation}\label{1.24}
\left\{
  \begin{array}{ll}
\widetilde{F}^{\lambda\mu}\widetilde{F}_{\lambda\mu}=-2g_{ij}\widetilde{F}^{0i}\widetilde{F}^{0j}+g^{ik}g^{jl}\widetilde{F}_{ij}\widetilde{F}_{kl}
;\\
  \widetilde{F}_{i\lambda}\widetilde{F}_{j}^{\;\;\lambda}=-g_{ik}g_{jl}\widetilde{F}^{0k}\widetilde{F}^{0l}+g^{kl}\widetilde{F}_{ik}\widetilde{F}_{jl};\\
  \widetilde{F}_{0\lambda}\widetilde{F}_{j}^{\;\;\lambda}
  =-\widetilde{F}_{jk}\widetilde{F}^{0k}; \\
  \widetilde{F}_{0\lambda}\widetilde{F}_{0}^{\;\;\lambda}
  =g_{ij}\widetilde{F}^{0i}\widetilde{F}^{0j};
  \end{array}
\right.
\end{equation}
from where we deduce, using (\ref{1.6}), that:
\begin{numcases}
\strut
\widetilde{\tau}_{00}=\frac{1}{2}g_{ij}\widetilde{F}^{0i}\widetilde{F}^{0j}+\frac{1}{4}g^{ik}g^{jl}\widetilde{F}_{kl}\widetilde{F}_{ij} \nonumber \\
 \widetilde{\tau}_{0j}= -\widetilde{F}^{0k}\widetilde{F}_{jk}  \nonumber\\
 \widetilde{\tau}_{ij}=(\frac{1}{2}g_{ij}g_{kl}-g_{ik}g_{jl})\widetilde{F}^{0k}\widetilde{F}^{0l}-\frac{1}{4}g_{ij}g^{km}g^{nl}\widetilde{F}_{kl}\widetilde{F}_{mn}
   +g^{kl}\widetilde{F}_{ik}\widetilde{F}_{jl}.\label{1.25}
\end{numcases}
    \item[$\bullet$]Next, to derive the equations for the scalar
    field $\Phi$ and the matter density $\rho$, we use the
    conservation laws:
\begin{equation}\label{1.26}
    \widetilde{\nabla}_{\alpha}\widetilde{T}^{\alpha
\beta}+\widetilde{\nabla}_{\alpha}\widetilde{\tau}^{\alpha
\beta}+\widetilde{\nabla}_{\alpha}(\rho
\widetilde{u}^{\alpha}\widetilde{u}^{\beta})=0.
\end{equation}
A direct calculation using (\ref{1.5}), (\ref{1.6}), (\ref{1.21})
and the Maxwell equation (\ref{1.4}) gives:
\begin{equation}\label{1.27}
\widetilde{\nabla}_{\alpha}\widetilde{T}^{\alpha\beta}=
\widetilde{\nabla}^{\beta}\Phi(\widetilde{\nabla}_{\alpha}\widetilde{\nabla}^{\alpha}\Phi-m^2\Phi);
\;\;\widetilde{\nabla}_{\alpha}\widetilde{\tau}^{\alpha\beta}=
\widetilde{F}^{\beta}_{\;\;\lambda}\widetilde{\nabla}_{\alpha}\widetilde{F}^{\alpha\lambda}
\end{equation}
where $\widetilde{\nabla}_{\alpha}\widetilde{\nabla}^{\alpha}$ often
denoted $\Box_{\widetilde{g}}$ is the d'Alembertian or the wave
operator.\\
Now (\ref{1.26}) and (\ref{1.27}) give, using the Maxwell equation
(\ref{1.3}):
\begin{equation}\label{1.28}
\widetilde{\nabla}^{\beta}\Phi(\widetilde{\nabla}_{\alpha}\widetilde{\nabla}^{\alpha}\Phi-m^2\Phi)+
4\pi e
\widetilde{F}^{\beta}_{\;\;\lambda}\widetilde{u}^{\lambda}+\widetilde{u}^{\beta}\widetilde{\nabla}_{\alpha}(\rho
\widetilde{u}^{\alpha})+(\rho
\widetilde{u}^{\alpha})\widetilde{\nabla}_{\alpha}\widetilde{u}^{\beta}=0
\end{equation}
(\ref{1.28}), reduces using (\ref{1.8}) to:
\begin{equation}\label{1.29}
\widetilde{\nabla}^{\beta}\Phi(\widetilde{\nabla}_{\alpha}\widetilde{\nabla}^{\alpha}\Phi-m^2\Phi)
+\widetilde{u}^{\beta}\widetilde{\nabla}_{\alpha}(\rho
\widetilde{u}^{\alpha})=0.
\end{equation}
But it is easily seen that $\widetilde{\nabla}^{i}\Phi=0$; so
(\ref{1.29}) gives for $\beta=i$,
$\widetilde{u}^{i}\widetilde{\nabla}_{\alpha}(\rho
\widetilde{u}^{\alpha})=0 $ which is satisfied for every $i=1,\,
2,\, 3$ if:
\begin{equation}\label{1.30}
\widetilde{\nabla}_{\alpha}(\rho \widetilde{u}^{\alpha})=0.
\end{equation}
Now (\ref{1.29}) gives, since $\widetilde{\nabla}^{i}\Phi=0$ and
using (\ref{1.30}):
\begin{equation}\label{1.31}
\widetilde{\nabla}^{0}\Phi(\widetilde{\nabla}_{\alpha}\widetilde{\nabla}^{\alpha}\Phi-m^2\Phi)=0.
\end{equation}
It then appears that the conservation laws (\ref{1.26}) will be
satisfied if equations (\ref{1.30}) and (\ref{1.31}) in $\rho$ and
$\Phi$ are.
    \item[$\bullet$]Now a direct calculation shows that equation
    (\ref{1.30}) in $\rho$ writes: $$\dot{\rho}=-\Big(\frac{\dot{\widetilde{u}}^{0}}{\widetilde{u}^{0}}
    -k^{i}_{i}+C_{ij}^{i}\frac{\widetilde{u}^{j}}{\widetilde{u^{0}}}\Big)\rho,$$
    which solves at once over $[0, t]$, $t>0$, to give:
\begin{equation}\label{1.32}
 \rho=\frac{\rho(0)
 \widetilde{u}^{0}(0)}{\widetilde{u}^{0}}\exp\Big[\int_{0}^{t}(k^{i}_{i}
 -C_{ij}^{i}\frac{\widetilde{u}^{j}}{\widetilde{u^{0}}})(s)ds\Big].
\end{equation}
(\ref{1.32}) shows that $(\rho(0)>0)\Longrightarrow (\rho>0) $. In
what follows we set:
\begin{equation}\label{1.32bis}
\rho(0)>0.
\end{equation}
Next, it is easily seen, using (\ref{1.21}), that
$$\widetilde{\nabla}_{\alpha}\widetilde{\nabla}^{\alpha}\Phi=-\ddot{\Phi}+H\dot{\Phi},$$
so that equation (\ref{1.31}) in $\Phi$ can be written, using
$\widetilde{\nabla}^{0}\Phi=\widetilde{g}^{00}\widetilde{\nabla}_{0}\Phi=-\dot{\Phi}$:
\begin{equation}\label{1.33}
\dot{\Phi}\ddot{\Phi}-H(\dot{\Phi})^2+m^2\Phi\dot{\Phi}=0.
\end{equation}
To study this non linear second order equation in $\Phi$, we set:
\begin{equation}\label{1.34}
U=\frac{1}{2}(\dot{\Phi})^2.
\end{equation}
We choose to look for a non-decreasing scalar field $\Phi$, which
means $\dot{\Phi}\geq0$; (\ref{1.34}) then gives:
\begin{equation}\label{1.35}
    \dot{\Phi}=\sqrt{2U}.
\end{equation}
    \item[$\bullet$] Next, to derive equation in $e$, we use the
    Maxwell equation (\ref{1.3}), which gives, using the identity
    $\widetilde{\nabla}_{\alpha}\widetilde{\nabla}_{\beta}\widetilde{F}^{\alpha\beta}=0$
     and the expression (\ref{1.7}) of $\widetilde{J}^{\beta}$:
     \begin{equation}\label{1.36}
\widetilde{\nabla}_{\alpha}(e \widetilde{u}^{\alpha})=0
     \end{equation}
(\ref{1.36}), yields, similarly to (\ref{1.30}):
\begin{equation}\label{1.36bis}
 e=\frac{e(0)
 \widetilde{u}^{0}(0)}{\widetilde{u}^{0}}\exp\Big[\int_{0}^{t}(k^{i}_{i}
 -C_{ij}^{i}\frac{\widetilde{u}^{j}}{\widetilde{u^{0}}})(s)ds\Big]
\end{equation}
by (\ref{1.36bis}): $(e(0)\geq 0)\Longrightarrow e\geq 0$. In what
follows we set:
\begin{equation}\label{1.36ter}
    e(0)\geq 0.
\end{equation}

    \item[$\bullet$]We now write in details the Maxwell equations
    (\ref{1.3}), (\ref{1.4}) in $\widetilde{F}$; (\ref{1.3})
    gives for $\beta=i$ and $\beta=0$, and using the expression
    (\ref{1.7}) of $\widetilde{J}^{\beta}$
\begin{numcases}
\strut
 \widetilde{\nabla}_{\alpha}\widetilde{F}^{\alpha i}=4 \pi
 e\widetilde{u}^{i}\label{1.37}\\
\widetilde{\nabla}_{\alpha}\widetilde{F}^{\alpha 0}=4 \pi
 e\widetilde{u}^{0}.\label{1.38}
\end{numcases}
\begin{itemize}
    \item[-] Applying formulae (\ref{1.21}) yields:
    \begin{equation}\label{1.39}
\widetilde{\nabla}_{\alpha}\widetilde{F}^{\alpha
i}=\widetilde{e}_{\alpha}\widetilde{F}^{\alpha
i}+\widetilde{\gamma}^{\alpha}_{\alpha\lambda}\widetilde{F}^{\lambda
i}+\widetilde{\gamma}^{i}_{\alpha\lambda}\widetilde{F}^{\alpha\lambda}.
    \end{equation}
But by a direct calculation using (\ref{1.19}), (\ref{1.20}) the
properties\\ $C^{k}_{ij}=-C^{k}_{ji}$,
$\widetilde{F}^{ij}=-\widetilde{F}^{ji}$ and since
$$\widetilde{e}_{\alpha}\widetilde{F}^{\alpha
i}=\partial_{0}\widetilde{F}^{0 i}=\dot{\widetilde{F}}^{0 i}$$ we
deduce from (\ref{1.37}) and (\ref{1.39}), the equation in
$\widetilde{F}^{0i}$:
\begin{equation}\label{1.40}
\dot{\widetilde{F}}^{0i}=H\widetilde{F}^{0i}
-C^{j}_{jk}\widetilde{F}^{ki}-\frac{1}{2}C^{i}_{jk}\widetilde{F}^{jk}+4
\pi e\widetilde{u}^{i}.
\end{equation}
Now concerning (\ref{1.38}), we have, using (\ref{1.19}) and
(\ref{1.20}):\\
$\widetilde{\nabla}_{\alpha}\widetilde{F}^{\alpha0}=C^{i}_{ik}\widetilde{F}^{k0}$,
so that (\ref{1.38}) gives the constraints equation:
\begin{equation}\label{1.41}
C^{i}_{ik}\widetilde{F}^{0k}+4 \pi e\widetilde{u}^{0}=0.
\end{equation}
\item[-] Next, observe that the Maxwell equations (\ref{1.4})
    split into the two sets:
    \begin{equation}\label{1.42}
\widetilde{\nabla}_0\widetilde{F}_{ij}+\widetilde{\nabla}_i\widetilde{F}_{j0}+\widetilde{\nabla}_j\widetilde{F}_{0i}=0;
\;\;
\widetilde{\nabla}_i\widetilde{F}_{jk}+\widetilde{\nabla}_j\widetilde{F}_{ki}+\widetilde{\nabla}_k\widetilde{F}_{ij}=0
    \end{equation}
the first equation (\ref{1.42}) gives the equation in
$\widetilde{F}_{ij}$:
\begin{equation}\label{1.43}
    \dot{\widetilde{F}}_{ij}=-C^{k}_{ij}\widetilde{F}_{0k}
\end{equation}
and the second equation (\ref{1.42}) gives the constraint equation:
\begin{equation}\label{1.44}
 C^{l}_{ij}\widetilde{F}_{kl}+C^{l}_{ki}\widetilde{F}_{jl}+C^{l}_{jk}\widetilde{F}_{il}:=C^{l}_{[ij}F_{k]l}=0
\end{equation}
\end{itemize}

    \item[$\bullet$]Since all the indices are now fixed, we set from
    now, to simplify notations:
\begin{equation}\label{1.45}
\left\{
  \begin{array}{ll}
   \widetilde{F}^{0i}=E^{i};\;\;\widetilde{F}_{ij}=F_{ij};\;\;\widetilde{T}_{00}=T_{00};\;\;\widetilde{T}_{0i}=T_{0i};\;\;\widetilde{T}_{ij}=T_{ij} & \hbox{} \\
\widetilde{\tau}_{00}=\tau_{00};\;\;\widetilde{\tau}_{0i}=\tau_{0i};\;\;\widetilde{\tau}_{ij}=\tau_{ij};\;\;\widetilde{u}^{0}=u^{0};\;\;\widetilde{u}^{i}=u^{i}&
\hbox{}
  \end{array}
\right.
\end{equation}
    \item[$\bullet$]Finally, since the system in $(u^{i})$ is given
    by (\ref{1.8}) for $\beta=i$, it remains to explicit
    the Einstein equations in $(g_{ij}, k_{ij})$; the uncharged case
    $F=0$ is given by classical equations. The equations in the
    charged case $F\neq 0$ were set up by \cite{16}, and we
    easily adapt to the present case. We are then led to the following
   \textit{evolution system} in $(g_{ij}, k_{ij}, E^{i}, F_{ij}, u^{i}, \Phi, U, \rho, e )$:
\begin{numcases}
\strut
 \dot{g}_{ij}= -2k_{ij} \label{1.46} \\
   \dot{k}_{ij}=R_{ij}+Hk_{ij}-2k^{l}_{j}k_{il}-8\pi(T_{ij}+
\tau_{ij}+\rho  u_{i} u_{j})\nonumber \\
 +4\pi \Big[-T_{00}-\rho u^{2}_{0}+g^{lm}(T_{lm}+\rho u_{l}u_{m})\Big]g_{ij} -\Lambda g_{ij}\label{1.47}  \\
\dot{E}^{i} =HE^{i}-\Big(C^{j}_{jk}g^{kl}g^{im}+\frac{1}{2}C^{i}_{jk}g^{jl}g^{km}\Big)F_{lm}+4\pi e u^{i} \label{1.48}  \\
\dot{F}_{ij}= C^{k}_{ij}g_{kl}E^{l} \label{1.49} \\
\dot{u}^{i} =2k_{j}^{i}u^{j}-\gamma^{i}_{jk}\frac{u^{j}u^{k}}{u^{0}}-4\pi\frac{e}{\rho}E^{i}+4\pi\frac{e}{\rho}\frac{g^{il}F_{jl}u^{j}}{u^{0}} \label{1.50}  \\
\dot{\Phi}   =  \sqrt{2U} \label{1.51} \\
 \dot{U }  = 2H U - m^2\Phi \sqrt{2U}\label{1.52}  \\
 \dot{\rho }  =-\Big[ k_{ij}\frac{u^{i}u^{j}}{(u^{0})^2}+C^{i}_{ij}\frac{u^{j}}{u^{0}}-k^{i}_{i}\Big]\rho +4\pi e
 g_{ij}\frac{E^{j}u^{i}}{(u^{0})^2}\label{1.53}\\
\dot{e }  =-\Big[
k_{ij}\frac{u^{i}u^{j}}{(u^{0})^2}+C^{i}_{ij}\frac{u^{j}}{u^{0}}-k^{i}_{i}\Big]e
 +4\pi g_{ij}\frac{E^{j}u^{i}}{(u^{0})^2}\frac{e^2}{\rho}
\label{1.54}
\end{numcases}
where:
\begin{enumerate}
    \item[-]In (\ref{1.47}) $R_{ij}$ is the Ricci tensor
    associated to $g_{ij}$ and given see \cite{2} by:
    \begin{equation}\label{1.55}
R_{ij}=\gamma^{l}_{lm}\gamma^{m}_{ji}-\gamma^{m}_{jl}\gamma^{l}_{mi}-C^{l}_{mj}\gamma^{m}_{li};
    \end{equation}
    \item[-]In (\ref{1.50}) $\gamma^{i}_{jk}$ is given by
    (\ref{1.19});
    \item[-]The equation (\ref{1.52}) in U is given by
    (\ref{1.33}) using the change of variable (\ref{1.34})
    which provides in the same time equation (\ref{1.51}) in
    $\Phi$ following the choice $\dot{\Phi}\geq 0$;
    \item[-]Equation (\ref{1.50}) in $u^{i}$ is given by
    (\ref{1.8}) for $\beta=i$;
    \item[-]Equations (\ref{1.53}) and (\ref{1.54}) in $\rho$
    and $e$ are given respectively by (\ref{1.30}) and
    (\ref{1.36}) in which $\dot{u}^{0}$ is provided by
    (\ref{1.8}) for $\beta=0$, which gives:
\begin{equation*}
\dot{u}^{0}=k_{ij}\frac{u^{i}u^{j}}{u^{0}}-4 \pi
    \frac{e}{\rho}g_{ij}\frac{u^{i}E^{j}}{u^{0}}.
\end{equation*}
    Finally, (\ref{1.46}) is provided by (\ref{1.9}) and (\ref{1.47})
    is a direct adaptation of (\ref{1.14}) in \cite{16} to the present
    case.
\end{enumerate}

\item[$\bullet$]Next we have the following set of constraints
equations:
\begin{numcases}
\strut
 R-k_{ij}k^{ij}+H^{2}=16\pi(T_{00}+\tau_{00}+\rho u_{0}^{2})+2\Lambda \label{1.56} \\
      \nabla^{i}k_{ij}=-8\pi(T_{0j}+\tau_{0j}+\rho u_{0}u_{j}) \label{1.57}\\
    C^{l}_{[ij}F_{k]l}:=C_{ij}^{l}F_{kl}+C_{ki}^{l}F_{jl}+C_{jk}^{l}F_{il}=0 \label{1.58}\\
C_{ik}^{i}E^{k}+4\pi eu^{0} = 0\label{1.59}
\end{numcases}
where (\ref{1.58}) and (\ref{1.59}) are given by  (\ref{1.44}) and
(\ref{1.41}) whereas (\ref{1.56}) (called the Hamiltonian
constraint) in which $R=g^{ij}R_{ij}$ and (\ref{1.57}) are easily
deduced from (1.20) and (1.21) in \cite{16}.
\end{enumerate}
\section{Study of constraints and mean curvature: the Cauchy problem}
\label{para2} We first set up the evolution of the different
quantities involved. We prove:
\begin{lemma}\label{llem1}
\hspace*{2cm}\\
If the evolution system is satisfied, then we have:
\begin{numcases}
\strut \frac{d H}{dt} =R+H^2+4\pi g^{ij}(T_{ij}+\rho
u_{i}u_{j})-12\pi(T_{00}+\rho u_0^2)-8\pi \tau_{00}-3\Lambda.\label{2.1}\\
\frac{dH^{2}}{dt} =2H\Big[R+H^2-3\Lambda+4\pi g^{ij}(T_{ij}+\rho u_{i} u_{j})-12\pi(T_{00}+\rho u_0^{2})-8\pi \tau_{00}\Big]. \nonumber\\
\label{2.2}\\
\frac{d}{dt}(k^{ij}k_{ij})=2k^{ij}R_{ij}
+2H(k^{ij}k_{ij}-\Lambda)-16\pi k^{ij}(T_{ij}+\tau_{ij}+\rho u_{i}u_{j})  \nonumber\\
 \hspace*{2cm}            +8\pi H[-T_{00}-\rho u_0^2+g^{lm}(T_{lm}+\rho u_lu_m)].  \label{2.3}\\
\frac{d}{dt}(T_{00}+\tau_{00}+\rho u_{0}^2)=H(T_{00}+\tau_{00}+\rho
u^{2}_{0})+k^{ij}(T_{ij}+\tau_{ij}+\rho u_{i}u_{j})\nonumber\\
\hspace*{4cm} -\nabla_{l}(T^{0l}+\tau^{0l}+\rho u^{0}u^{l}). \label{2.4}\\
\frac{d}{dt}(T^{0j}+\tau^{0j}+\rho u^{0}u^{j})=
  H(T^{0j}+\tau^{0j}+\rho u_{0}u_{j})+2k^{j}_i(T^{0i}+\tau^{0i}+\rho u^{i}u^{0})\nonumber\\
   \hspace*{4cm}     -\nabla_{i}(T^{ij}+\tau^{ij}+\rho
 u^{i}u^{j}).\label{2.5}\\
 \frac{d }{dt}R=2k^{ij}R_{ij}-2\nabla_{l}\nabla_{i}k^{il}.\label{2.6}
\end{numcases}
\end{lemma}
\begin{proof}
\hspace*{2cm}\\
See Appendix $A_{2}$.
\end{proof}
\begin{lemma}\label{llem2}
\hspace*{2cm}\\
Set:
\begin{numcases}
\strut
A=R-k_{ij}k^{ij}+H^{2}-16\pi(T_{00}+\tau_{00}+\rho u_{0}^{2})-2\Lambda; \nonumber \\
   A_{j}=\nabla^{i}k_{ij}+8\pi(T_{0j}+\tau_{0j}+\rho u_{0}u_{j});  \nonumber\\
   A_{ijk}=C_{ij}^{l}F_{kl}+C_{ki}^{l}F_{jl}+C_{jk}^{l}F_{il};  \nonumber\\
   B=C_{ik}^{i}E^{k}+4 \pi eu^{0}\nonumber
\end{numcases}
then
\begin{numcases}
\strut
\frac{d A}{dt}=2H A+2g^{jm}\gamma_{mj}^{k} A_{k}\label{2.7}  \\
\frac{d A_j}{dt} = HA_j \label{2.8} \\
\frac{d A_{ijk}}{dt}=0 \label{2.9} \\
\frac{d B}{dt}=H B.\label{2.10}
\end{numcases}
\end{lemma}
\begin{proof}
\hspace*{2cm}\\
\begin{enumerate}
    \item[1.)](\ref{2.7}) is a consequence of (\ref{2.2}),
    (\ref{2.3}), (\ref{2.4}) and (\ref{2.6});
    \item[2.)](\ref{2.9}) is a consequence of the evolution
    equation (\ref{1.49}) in $F_{ij}$ and the Jacobi polynomial
    relation $C^{l}_{[ij}C_{k]}^{m}=0$;
    \item[3.)]To obtain (\ref{2.10}), first set:
    \begin{equation}\label{2.11}
C^{\beta}=\widetilde{\nabla}_{\alpha}\widetilde{F}^{\alpha\beta}-4
\pi e \widetilde{u}^{\beta}
    \end{equation}
for $\beta=i$, given the evolution equation (\ref{1.37}) in
$\widetilde{F}^{0i}=E^{i}$ which is also written in the form
(\ref{1.40}) or (\ref{1.48}), we have
\begin{equation}\label{2.12}
    C^{i}=0.
\end{equation}
Next, due to the identity
$\widetilde{\nabla}_{\alpha}\widetilde{\nabla}_{\beta}\widetilde{F}^{\alpha\beta}=0$
and (\ref{1.36}) i.e
$\widetilde{\nabla}_{\beta}(e\widetilde{u}^{\beta})=0$, we have
\begin{equation}\label{2.13}
\widetilde{\nabla}_{\beta}C^{\beta}=0;
\end{equation}
then by a direct calculation using  (\ref{1.19}), (\ref{1.21}) and
(\ref{2.12}), (\ref{2.13}) gives
\begin{equation}\label{2.14}
    \partial_{t}C^{0}-HC^{0}=0
\end{equation}
But (\ref{2.11}) gives:
\begin{equation*}
C^{0}=\widetilde{\nabla}_{\alpha}\widetilde{F}^{\alpha0}-4 \pi e
\widetilde{u}^{0}=-C_{ik}^{i}\widetilde{F}^{0k}-4 \pi e
\widetilde{u}^{0}=-C_{ik}^{i}E^{k}-4 \pi e \widetilde{u}^{0}=-B.
\end{equation*}
 So
(\ref{2.10}) follows from (\ref{2.14}).
    \item[4.)]For the proof of (\ref{2.8}) see Appendix $A_{3}$.
\end{enumerate}
\end{proof}
\begin{proposition}\label{prop1}
\hspace*{2cm}\\
The constraints (\ref{1.56}), (\ref{1.57}), (\ref{1.58}) and
(\ref{1.59}) are satisfied in the whole domain of existence of the
solutions of the evolution system (\ref{1.46}), (\ref{1.47}),
(\ref{1.48}), (\ref{1.49}), (\ref{1.50}), (\ref{1.51}),
(\ref{1.52}), (\ref{1.53}), (\ref{1.54}), if and only if, they are
satisfied for $t=0$.
\end{proposition}
\begin{proof}
\hspace*{2cm}\\
Integration of (\ref{2.7}), (\ref{2.8}), (\ref{2.10}) over $[0, t]$,
$t>0$ gives:
\begin{equation*}
A_{j}=A_{j}(0)\exp(\int_{0}^{t}H ds);\;\; A_{ijk}=A_{ijk}(0);\;\;
B=B(0)\exp(\int_{0}^{t}H ds);
\end{equation*}
hence $(A_{j}(0)=A_{ijk}(0)=B(0)=0
)\Longleftrightarrow(A_{j}=A_{ijk}=B=0 ) $.\\
Now setting $A_{k}=0$ in (\ref{2.7}) gives $\dot{A}=2HA$ which
integrates over $[0, t]$, $t>0$ to give $A=A(0)\exp(\int_{0}^{t}2H
ds)$. Hence $(A(0)=0)\Longleftrightarrow(A=0)$ and
proposition~\ref{prop1} follows.
\end{proof}
We suppose from now on that the constraints (\ref{1.56}),
(\ref{1.57}), (\ref{1.58}) and (\ref{1.59}) are satisfied for $t=0$.
As consequence, we can use the constraints which can now be
considered as \textit{properties} of the solution of the evolution
system.

 We now prove an important theorem on the mean
curvature of the solutions of the evolution system.
\begin{theorem}\label{ttheo1}
\hspace*{2cm}\\
Let $\Phi(0)>0$, $\Lambda> -4\pi m^2(\Phi(0))^2$ be given, and
suppose $H(0)=(g^{ij}k_{ij})(0)<0$ then H is uniformly bounded and
we have:
\begin{equation}\label{2.15}
    H(0)\leq H\leq -\sqrt{3\Lambda+12\pi m^2(\Phi(0))^2}.
\end{equation}
\end{theorem}
\begin{proof}
\hspace*{2cm}\\
Denote by $\sigma_{ij}$ the traceless tensor associated to $k_{ij}$,
i.e
\begin{equation}\label{2.16}
k_{ij}=\frac{H}{3}g_{ij}+\sigma_{ij}.
\end{equation}
A direct calculation gives:
\begin{equation}\label{2.17}
    k^{ij}k_{ij}=\frac{1}{3}H^2+\sigma_{ij}\sigma^{ij}.
\end{equation}
We now use the evolution equation (\ref{2.1}) in H, in which we use
the Hamiltonian constraint (\ref{1.56}) to express the quantity
$R+H^2$,  (\ref{1.23}) and (\ref{1.34}) to express $T_{00}$,
$T_{ij}$ in terms of $\Phi$ and U, and finally (\ref{2.17}) to
express $k_{ij}k^{ij}$, to obtain:
\begin{equation}\label{2.18}
\frac{dH}{dt}=\frac{H^2}{3}-\Lambda-4\pi
m^2\Phi^2+\sigma_{ij}\sigma^{ij}+16\pi U +4\pi g^{ij} u_{i}
u_{j}+4\pi\rho u_{0}^{2}+8\pi\tau_{00}.
\end{equation}
But since $ \sigma_{ij}\sigma^{ij}\geq 0$, $g^{ij} u_{i} u_{j}\geq
0$, $\tau_{00}\geq 0$ (see (\ref{1.25})), (\ref{2.18}) gives:
\begin{equation}\label{2.19}
\frac{dH}{dt}\geq \frac{H^{2}}{3}-\Lambda-4\pi m^2 \Phi^2
\end{equation}
Consider once more the Hamiltonian constraint (\ref{1.56}) which
gives, using (\ref{2.17}) to express $k_{ij}k^{ij}$ and (\ref{1.23})
to express $T_{00}$:
\begin{equation}\label{2.20}
 \frac{2}{3}H^2-2\Lambda-8\pi m^2\Phi^{2}=16\pi U+16\pi(\tau_{00}+\rho
 u^2_0)+\sigma_{ij}\sigma^{ij}-R.
\end{equation}
But it proved in \cite{10}, \cite{26}, that for the models under
consideration, we always have: $R\leq 0$; (\ref{2.20}) then gives:
\begin{equation}\label{2.21}
\frac{H^2}{3}-\Lambda-4\pi m^2\Phi^{2}\geq 0;
\end{equation}
(\ref{2.19}) then implies:
\begin{equation}\label{2.22}
    \frac{dH}{dt}\geq 0
\end{equation}
so that H is non-decreasing. We also deduce from (\ref{2.21}) since
by (\ref{1.51}) we have $\dot{\Phi}\geq 0$, then
$\Phi\geq\Phi(0)>0$, that:
\begin{equation}\label{2.23}
    H^{2}\geq 3\Lambda +12 \pi m^2 (\Phi(0))^2.
\end{equation}
But by hypothesis, the r.h.s of (\ref{2.23}) is strictly positive.
So, since H is continuous, (\ref{2.23}) implies:
\begin{equation}\label{2.24}
H\leq -\sqrt{3\Lambda+12\pi m^2(\Phi(0))^2}\;\; or \;\; H\geq
\sqrt{3\Lambda+12\pi m^2(\Phi(0))^2}.
\end{equation}
Also by hypothesis, $H(0)<0$, then only the first inequality in
(\ref{2.24}) holds; moreover, (\ref{2.22}) implies $H\geq H(0)$ and
(\ref{2.15}) follows.
\end{proof}
\indent We now introduce the Cauchy or initial value problem, taking
into account (\ref{1.32bis}), (\ref{1.34}), (\ref{1.36ter}); let the
following quantities called \textit{initial data} be given:
\begin{equation*}
    \left\{
      \begin{array}{ll}
g^{0}=(g^{0}_{ij})&\hbox{\hspace*{-1.2cm}a positive definite $3\times3$ constant matrix;} \\
k^{0}=(k^{0}_{ij})&\hbox{\hspace*{-1.2cm}a symmetric  $3\times3$ constant matrix;} \\
F^{0}=(F_{ij}^{0})&\hbox{\hspace*{-1cm}an antisymmetric 3$\times$3 constant matrix;} \\
 u^{0}=(u^{(0),\;i}); &\hbox{\hspace*{-0.5cm}$E^{0}=(E^{0,\;i})$,\;\; constant vectors;} \\
 \Phi^{0}>0;  U^{0}>0; &\hbox{$\rho^{0}>0$;\;\; $e^{0}\geq 0$,\;\; real
   numbers};
      \end{array}
    \right.
\end{equation*}
we look for $g=(g_{ij})$, $k=(k_{ij})$, $E=(E^{i})$, $F=(F_{ij})$,
$u=(u^{i})$, $\Phi$, $U$, $\rho$, $e$ solutions of the evolution
system such that:
\begin{numcases}
\strut
g(0)=g^{0};\;   k(0)=k^{0};\;  E(0)=E^{0};\;  F(0)=F^{0};\; u(0)=u^{(0)} \nonumber\\
\Phi(0) =\Phi^{0};\; U(0)=U^{0};\; \rho(0)=\rho^{0};\; e(0)=e^{0}.
\label{2.255}
\end{numcases}
By Proposition~\ref{prop1}, the constraint equations (\ref{1.56}),
(\ref{1.57}), (\ref{1.58}) ,(\ref{1.59}) are satisfied if and only
if the initial data satisfy these constraints we call
\textit{initial constraints}.
In what follows, we consider that it is the case.\\
We end this section by the useful notion of \textit{relative norm}.
Define the norm of a $n\times n$ matrix A by:
\begin{equation*}
\|A\|=\sup\left\{
  \begin{array}{ll}
\frac{\|Ax\|}{\|x\|},\; x\in \mathbb{R}^{n}, \; x\neq 0
  \end{array}
\right\}.
\end{equation*}
If $A_1$ and $A_2$ are two symmetric matrices with $A_1$ positive
definite, define the norm of $A_2$ with respect to $A_1$ by:
\begin{equation*}
 \|A_{2}\|_{A_1}=\sup\left\{
  \begin{array}{ll}
\frac{\|A_{2}x\|}{\|A_{1}x\|},\; x\in \mathbb{R}^{n}, \; x\neq 0
  \end{array}
\right\}.
\end{equation*}
We have the following results proved in \cite{17}:
\begin{lemma}\label{lem3}
\begin{eqnarray}
  \|A_2\| &\leq &\|A_2\|_{A_1} \|A_1\| \label{2.25}. \\
   \|A_2\|_{A_1} &\leq& [tr(A^{-1}_1A_2A^{-1}_1A_2)]^{\frac{1}{2}}.\label{2.26}
\end{eqnarray}
\end{lemma}
\noindent
Now by setting $A_{1}=(g^{ij})$, $A_{2}=(a_{ij})$ a direct
calculation gives:
\begin{equation}\label{2.27}
tr(A^{-1}_1A_2A^{-1}_1A_2)=a^{ij}a_{ij}.
\end{equation}
We then deduce at once from (\ref{2.25}), (\ref{2.26}), (\ref{2.27})
that in these case:
\begin{equation}\label{2.28}
\|A_2\| \leq \|A_1\| (a^{ij}a_{ij})^{\frac{1}{2}}.
\end{equation}
Next, let $A=(a_{ij})$ be a $n\times n$ matrix; set $|A|=\sup
\{|a_{ij}|, \; i,j=1,...,n\}$. Then we have the following result,
from \cite{16}:
\begin{lemma}\label{lem4}
\hspace*{2cm}\\
Let $(u^{\alpha})=(u^{0}, u^{i})$ where $u^{0}$ and $u^{i}$ are link
by (\ref{1.11}). Then there exists a constant $C>0$ such that:
\begin{equation}\label{2.29}
   \Big|\frac{u^{i}}{u^{0}}\Big|\leq C |g|^{\frac{3}{2}}\;;\;\;
   |F^{0i}|\leq (g_{lm}F^{0l}F^{0m})^{\frac{1}{2}}|g|^{\frac{3}{2}}.
\end{equation}
\end{lemma}
\section{Local and global Existence of solutions}
\label{para3} We use an iterative scheme.
\subsection{Construction of the iterated sequence}
We adopt the notations introduced in paragraph~\ref{para2}. We
construct the sequence $v_{n}=(g_{n}, k_{n}, E_{n}, F_{n},  u_{n},
\Phi_{n}, U_{n}, \rho_{n}, e_{n})$,  $n\in\mathbb{N}$ as follows:
\begin{enumerate}
    \item[$\bullet$] Set $g_{0}=g^{0}$; $k_{0}=k^{0}$;
 $E_{0}=E^{0}$; $F_{0}=F^{0}$; $u_{0}=u^{(0)}$;
$\Phi_{0}=\Phi^{0}$; $U_{0}=U^{0}$; $\rho_{0}=\rho^{0}$;
$e_{0}=e^{0}$.
    \item[$\bullet$]If $g_n$, $k_n$, $E_n$, $F_n$, $u_n$, $\Phi_n$, $U_{n}$, $\rho_{n}$, $e_n$,
    are known: define $\tilde{T}_{n,\alpha\beta}$,
$\tilde{\tau}_{n,\alpha\beta}$ by substituting $\widetilde{g}$,
$\widetilde{F}^{0i}$, $\widetilde{F}_{ij}$,  $\Phi$ in the
expressions (\ref{1.5}), (\ref{1.6}) of $\tilde{T}_{\alpha\beta}$,
$\tilde{\tau}_{\alpha\beta}$ by $g_{n}$, $E_{n}$, $F_{n}$,
$\Phi_{n}$.
    \item[$\bullet$]Define $v_{n+1}=(g_{n+1}, k_{n+1}, E_{n+1}, F_{n+1},  u_{n+1},
\Phi_{n+1}, U_{n+1}, \rho_{n+1}, e_{n+1})$ as solution of the
\textit{linear} ordinary differential equations (o.d.e) obtained by
substituting $g$, $k$, $E$, $F$,  $u$, $\Phi$, $U$, $\rho$, $e$,
$T_{\alpha\beta}$,
 $\tau_{\alpha\beta}$ in the r.h.s of the evolution system
(\ref{1.46}) to (\ref{1.54}), by $g_{n}$, $k_{n}$, $E_{n}$, $F_{n}$,
$u_{n}$, $\phi_{n}$, $U_{n}$,
 $\rho_{n}$, $e_{n}$,  $T_{n,\alpha\beta}$,
 $\tau_{n,\alpha\beta}$.
\end{enumerate}
 It is very important to notice that, for
 every $n$ the initial data for the linear o.d.e's are the \textit{same
 initial data} $g^{0}$, $k^{0}$, $E^{0}$, $F^{0}$,
$u^{(0)}$, $\Phi^{0}$, $U^{0}$, $\rho^{0}$ and $e^{0}$. We obtain
this way a sequence $v_{n}=$($g_{n}$, $k_{n}$, $E_{n}$, $F_{n}$,
$u_{n}$, $\Phi_{n}$, $U_{n}$, $\rho_{n}$, $e_{n}$) defined in a
maximal interval $[0, T_{n}[$, $T_{n}> 0$.
\subsection{Boundedness of the iterated sequence}
\begin{proposition}\label{prop2}
\hspace*{2cm}\\
There exists $T>0$ \textbf{independent of $n$}, such that the
iterated sequence \\$v_{n}=$($g_{n}$, $k_{n}$, $E_{n}$, $F_{n}$,
$u_{n}$, $\Phi_{n}$, $U_{n}$, $\rho_{n}$, $e_{n}$) is defined and
uniformly bounded over $[0, T[$.
\end{proposition}
\begin{proof}
\hspace*{2cm}\\
Let $N\in \mathbb{N}$, $N>1$, be an integer. Suppose that we have,
for $n\leq N-1$, the inequalities:
\begin{numcases}
\strut
\nonumber|g_{n}-g_{0}|\leq  A_1 \; ; \; (detg_{n})^{-1}\leq  A_2 \; ; \;|k_{n}-k_{0}|\leq A_3\; ; \;  |E_{n}-E_{0}|\leq A_4 \\
\nonumber|F_{n}-F_{0}| \leq A_5  \; ;
\;|u_{n}-u_{0}|\leq A_6\; ; \;|\phi_{n}-\phi_{0}|\leq A_7  \\
| U_{n}-U_{0}| \leq A_8 \; ; \; |\rho_{n}-\rho_{0}| \leq A_9 \; ;
\;|e_{(n)}-e_{(0)}| \leq A_{10}\label{3.1}
\end{numcases}
where $ A_i > 0$, $i=1$, 2, 3, 4, 5, 6, 7, 8, 9, 10 are given
constants.
\\\indent
We are going to prove that one can choose the constants $A_{i}$ such
that (\ref{3.1}) still holds for $n=N$ on $[0, T[$, $T>0$,
sufficiently small. Notice that the expression of
$(g^{ij}_{n})=(g_{n, ij})^{-1}$ contains $(det g_{n})^{-1}$.
\begin{enumerate}
    \item[$\bullet$]Integrating over $[0, t]$, $t>0$, the linear
    o.d.e satisfied by: $g_{N}$, $k_{N}$, $E_{N}$, $F_{N}$,
$u_{N}$, $\Phi_{N}$, $U_{N}$, $\rho_{N}$ yields:
\begin{numcases}
\strut \nonumber|g_{N}-g_{0}| \leq B_1t \; ; \; |k_{N}-k_{0}|\leq
B_3t\;
; \;|E_{N}-E_{0}| \leq B_4t\;; \; |F_{N}-F_{0}|\leq B_5t \\
|\phi_{N}-\phi_{0}| \leq B_7t \; ;\;|U_{N}-U_{0}|\leq B_8t \;
;\;|\rho_{N}-\rho_{0}|\leq B_9t \label{3.2}
\end{numcases}
 where $B_i > 0$, $i=1$,  3, 4, 5,  7, 8, 9 are constants depending
 only on the $A_{i}$.\\
 We now study the cases of $(det g_{N})^{-1}$, $u_{N}$ and $e_{N}$.
    \item[$\bullet$]
    The iterated equation satisfied by $g_{N}$ writes, using
    (\ref{1.46}):
\begin{equation}\label{3.3}
    \dot{g}_{N,ij}=-2k_{N-1,ij}.
\end{equation}
Recall the formula:
\begin{equation}\label{3.4}
    \frac{d}{dt}[ln(detg_{N}]=g^{ij}_{N}\partial_{t} g_{N,ij};
\end{equation}
on the other hand we have:
\begin{equation}\label{3.5}
 \frac{d}{dt}[ln(detg_{N})]=(detg_{N})^{-1}\frac{d}{dt}(detg_{N})=-(detg_{N})\frac{d}{dt}(detg_{N})^{-1}.
\end{equation}
(\ref{3.5}) and (\ref{3.4}) then give, using (\ref{3.3}):
\begin{equation*}
    \frac{d}{dt}(detg_{N})^{-1}=(2g^{ij}_{N} k_{N-1,ij})(detg_{N})^{-1}
\end{equation*}
an o.d.e in $(detg_{N})^{-1}$ which integrate at once over $[0, t]$,
$t>0$ to give:
\begin{equation}\label{3.6}
    (detg_{N})^{-1}=(detg^0)^{-1}\exp \Big(\int^t_02(g^{ij}_{N} k_{N-1,ij})(s)ds\Big)
\end{equation}
Now (\ref{3.3}) which is analogous to (\ref{1.46}) shows that
$g_{N}$ and $k_{N-1}$ are the first and second fundamental forms of
a space-like hypersurface; so $g^{ij}_{N}
k_{N-1,ij}=k^{i}_{N-1,i}=tr(k_{N-1})$. We then deduce from
(\ref{3.6}) using (\ref{3.2}), that:
\begin{equation}\label{3.7}
(detg_{N})^{-1}\leq (detg^0)^{-1} \exp(C_{1}t)
\end{equation}
where $C_{1}> 0$ is a constant depending only on $A_{i}$ and
$|g^{0}|$, $|k^{0}|$. Hence, using (\ref{3.7}), it appears that if
we take in (\ref{3.1}): $A_{2}>(detg^0)^{-1} $ i.e
$(detg^0)A_{2}>1$, then given the continuity of $t\longmapsto
\exp(C_{1}t)$, we will have for $t$ sufficiently small
$(detg^0)A_{2}>\exp(C_{1}t)>1$. Then, there exits $t_{1}>0$ such
that, for $0<t<t_{1}$ we have, using (\ref{3.7}):
\begin{equation}\label{3.8}
(detg_{N})^{-1}\leq A_{2}.
\end{equation}
\item[$\bullet$]
Next, in order to have for $u_{N}$ and $e_{N}$ inequalities
analogous to (\ref{3.2}), we need to bound $\dfrac{1}{\rho_{n}}$
which appears in the iterated equations in $u_{n+1}$, $e_{n+1}$,
built from equations (\ref{1.50}) and (\ref{1.54}) in $u$ and $e$.
Now following the definition of the iterated sequence $(v_{n})$ and
given equation (\ref{1.53}) in $\rho$, $\rho_{n+1}$ satisfies, for
$0\leq n \leq N-1$
\begin{equation}\label{3.9}
\dot{\rho}_{n+1}=G_{n}
\end{equation}
where
\begin{equation}\label{3.10}
G_{n}=-\Big[ k_{n,
ij}\frac{u_{n}^{i}u_{n}^{j}}{(u_{n}^{0})^2}+C^{i}_{ij}\frac{u_{n}^{j}}{u_{n}^{0}}-k^{i}_{n,
i}\Big]\rho_{n} + 4 \pi e_{n} g_{n,
ij}\frac{u_{n}^{i}E_{n}^{j}}{(u_{n}^{0})^2}.
\end{equation}
\end{enumerate}
We have in the expression (\ref{3.10}) of $G_n$, using (\ref{2.29}):
\begin{equation*}
    \Big|\frac{u^{i}_{n}}{u^{0}_{n}}\Big| \leq C
    |g_{n}|^{\frac{3}{2}},\;\; and \;\;
    \frac{|u^{i}_{n}|}{(u^{0}_{n})^2}=\frac{|u^{i}_{n}|}{u^{0}_{n}}\frac{1}{u^{0}_{n}}
    \leq\frac{|u^{i}_{n}|}{u^{0}_{n}} \;\; since
    \;\;u^{0}_{n}\geq 1
\end{equation*}
Then, by  (\ref{3.1}), $G_n$ is bounded, i.e, $\exists\; C > 0$,
$|G_n|< C$; hence, $|\dot{\rho}_{n+1}|\leq C$ and this implies:
\begin{equation}\label{3.11}
\dot{\rho}_{n+1}\geq -C.
\end{equation}
Integrating (\ref{3.11}) over $[0, t]$, $t>0$ yields $\rho_{n+1}\geq
\rho^{0}-Ct$ but since $\rho^{0}> 0$, there exits $t_{2}>0$ such
that for $t\in[0, t_{2}] $ we have $Ct<\frac{\rho^{0}}{2}$; then
$$(0\leq t\leq t_{2})\Longrightarrow(\rho_{n+1}\geq
\frac{\rho^{0}}{2})$$ thus,
$\frac{1}{\rho_{n+1}}\leq\frac{2}{\rho^{0}}$ and the sequence
$$\Big(\frac{1}{\rho_{n}}\Big)$$ is bounded. Hence $u_{N}$ and $e_{N}$
also satisfy:
\begin{equation}\label{3.12}
    |u_{N}-u_{0}|\leq B_{6}t   ; \;\;       |e_{N}-e_{0}|\leq B_{10}t
\end{equation}
where as in (\ref{3.12}) $B_{6}$, $B_{10}$ are constant depending
only on the $A_{i}$. We then conclude that if $T>0$ is such that:
\begin{equation*}
    T< \inf(t_{1}, t_{2}),\;\; B_{i}T<A_{i}, \;\;i=1,2,...,10,
\end{equation*}
then, by (\ref{3.2}) and (\ref{3.12}) $v_{N}$ also satisfies
(\ref{3.1}) on $[0, T[$. Hence, the iterated sequence $(v_{n})$ is
uniformly bounded over $[0, T[$.
\end{proof}
\subsection{Local existence of solutions}
\begin{theorem}\label{ttheo2}
\hspace*{2cm}\\
The initial value problem for the Einstein-Maxwell-Scalar Field
system has a unique local solution.
\end{theorem}
\begin{proof}
\hspace*{2cm}\\
We prove that the iterated sequence $(v_{n})$ converges uniformly on
each bounded interval $[0, \delta]\subset[0, T]$, $\delta>0$. For
this purpose, we study the difference \\$v_{n+1}-v_{n}$. But given
the evolution equations (\ref{1.51}) and (\ref{1.52}) in $\Phi$ and
U, we will deal with the difference:
\begin{equation*}
    \sqrt{2U_{n+1}}-\sqrt{U_{n}}=\frac{2(U_{n+1}-U_{n})}{\sqrt{2U_{n+1}}+\sqrt{U_{n}}}.
\end{equation*}
We then need to show first of all that the sequence
$$\Big(\frac{1}{\sqrt{2 U_{n}}}\Big)$$ is uniformly bounded.
\begin{enumerate}
    \item[$\bullet$] By (\ref{1.52}), the iterated equation
    providing $U_{n+1}$ writes:
    \begin{equation}\label{3.14'}
       \dot{U}_{n+1}=2H_{n}U_{n}-m^2\Phi_{n}\sqrt{2U_{n}}.
    \end{equation}
But by Proposition~\ref{prop2}, there exits a constant $C>0$ such
that we have over $[0, T[$:
\begin{equation*}
|2H_{n}U_{n}-m^2\Phi_{n}\sqrt{2U_{n}}|\leq C
\end{equation*}
(\ref{3.14'}) then gives:
\begin{equation*}
    \frac{d U_{n+1}}{dt}\geq -C
\end{equation*}
and integrating over $[0, t]$, $0\leq t<T$ yields:
\begin{equation*}
U_{n+1}\geq U^{0}-Ct.
\end{equation*}
Recall that $U^{0}>0$. Then taking $t$ sufficiently small such that
$Ct\leq \frac{U^{0}}{2}$ we have $U_{n+1}\geq \frac{U_{0}}{2}$. Then
$$ \frac{1}{\sqrt{2U_{n+1}}} \leq \frac{1}{\sqrt{U^{0}}}$$ which shows that
$\Big(\frac{1}{\sqrt{2 U_{n}}}\Big)$ is uniformly bounded over $[0,
T[$, $T>0$ small enough.

    \item[$\bullet$]Since $(\frac{1}{\rho_{n}})$ is also bounded,
    taking the difference between two consecutive iterated
    equations, we deduce from the evolution system, using\\
    $v_{n}(0)=v^{0}$, $\forall n$, that there exits a constant
    $C_{2}>0$ such that:
\begin{eqnarray}
\nonumber & &\hspace*{-0.5cm}|g_{n+1}-g_{n}|+|k_{n+1}-k_{n}|+|E_{n+1}-E_{n}|+|F_{n+1}-F_{n}|+|u_{n+1}-u_{n}|+ \\
\nonumber&&|\Phi_{n+1}-\Phi_{n}|+|U_{n+1}-U_{n}|+|\rho_{n+1}-\rho_{n}|
\nonumber   +|e_{n+1}-e_{n}|\leq \\ &&
\nonumber    C_2\int_0^{t}\Big(|g_{n}-g_{n-1}|+|k_{n}-k_{n-1}|+|E_{n}-E_{n-1}|+|F_{n}-F_{n-1}|+  \\
\nonumber   & &|u_{n}-u_{n-1}|+|\Phi_{n}-\Phi_{n-1}|+|U_{n}-U_{n-1}|+|\rho_{n}-\rho_{n-1}|+  \\
            & &|e_{n}-e_{n-1}| \Big)(s) ds.  \label{3.13}
\end{eqnarray}
So, if we set:
\begin{eqnarray}
\nonumber
\alpha_{n}&=&|g_{n+1}-g_{n}|+|k_{n+1}-k_{n}|+|F_{n+1}-F_{n}|+|E_{n+1}-E_{n}|+
|u_{n+1}-u_{n}|+\\
            &&|\Phi_{n+1}-\Phi_{n}|+|U_{n+1}-U_{n}|+|\rho_{n+1}-\rho_{n}|+|e_{n+1}-e_{n}| \label{3.14}
\end{eqnarray}
then (\ref{3.13}) shows that we have:
\begin{equation}\label{3.15}
    \alpha_{n}(t)\leq C_{2}\int_{0}^{t} \alpha_{n-1}(s)ds.
\end{equation}

(\ref{3.15}) gives, by an immediate induction on $n$:
\begin{equation}\label{3.16}
\alpha_{n}(t)\leq \|\alpha_{2}\|_{\infty}\frac{(C_{2
t})^{n-2}}{(n-2)!}.
\end{equation}
But since the series $\sum\limits_{n=0}^{+\infty}\frac{C^{n}}{n!}$
is convergent we have necessarily $\alpha_{n}(t) \longrightarrow 0$
as $n \longrightarrow +\infty$. Definition (\ref{3.14}) of
$\alpha_{n}$ then shows that each of the sequences $(g_{n})$,
$(k_{n})$, $(E_{n})$, $(F_{n})$,  $(u_{n})$, $(\Phi_{n})$,
$(U_{n})$, $(\rho_{n})$, $(e_{n})$ converges uniformly on each
bounded interval $[0, \delta]$, $0<\delta<T$ and that their
respective limits denoted: $g$, $k$,  $E$, $F$, $u$, $\Phi$, $U$,
$\rho$ and $e$ are continuous function of $t$. From the iterated
equations, it appears immediately that there exits a constant
$C_{3}>0$ such that:
\begin{equation}\label{3.17}
    \Big|\frac{d v_{n+1}}{dt} -\frac{d v_{n}}{dt}\Big|\leq C_{3}|v_{n+1}-v_{n}|
\end{equation}
where we set, for $v=(v^{i})$, $|v|=\sum|v^{i}|$. The convergence of
$(v_{n})$ then implies, given (\ref{3.17}), the convergence of
$$\Big(\frac{d v_{n}}{dt}\Big)$$ and hence, that each of the sequences $(\dot{g}_{n})$,
$(\dot{k}_{n})$, $(\dot{E}_{n})$, $(\dot{F}_{n})$,  $(\dot{u}_{n})$,
$(\dot{\Phi}_{n})$, $(\dot{U}_{n})$, $(\dot{\rho}_{n})$,
$(\dot{e}_{n})$ converges uniformly on each interval $[0, \delta]$,
$0<\delta<T$. Consequently the limit functions $g$, $k$,  $E$, $F$,
$u$, $\Phi$, $U$, $\rho$ and $e$ are of class $\mathcal{C}^{1}$ and
$v:=(g, k,  E, F, u, \Phi, U, \rho, e)$ is a local solution of the
coupled Einstein-Maxwell-Massive Scalar Field system.
\\\indent Now
to prove the \textit{uniqueness} of the solution, suppose $v_{1}$
and $v_{2}$ are two solutions of the Cauchy problem , with the
\textit{same initial data}. Then, defining $\alpha(t)$ the same way
as $\alpha_{n}$ (see (\ref{3.14})) for the difference
$|v_{1}-v_{2}|$, leads, using the evolution system to:
\begin{equation*}
    \alpha(t)\leq \int_{0}^{t}\alpha(s)ds
\end{equation*}
which gives by Gronwall Lemma $\alpha=0$; hence $v_{1}=v_{2}$ and
uniqueness follows.
\end{enumerate}
\end{proof}
\subsection{Global existence of solutions}
We prove:
\begin{theorem}\label{ttheo33}
\hspace*{2cm}\\
Let $\Lambda > -4 \pi m^{2}(\Phi^{0})^2 $ be given and suppose
$H^{0}:=(g^{0})^{ij}k^{0}_{ij}<0$. Then the initial value problem
for the Einstein-Maxwell-Massive Scalar Field system has a unique
global solution defined all over the interval $[0, +\infty[$.
\end{theorem}
\begin{proof}
\hspace*{2cm}\\
Following the standard theory on the first order differential
systems, it will be enough if we could prove, given the evolution
system (\ref{1.46}) to (\ref{1.54}) that, if each of the functions:
$|g|$, $|k|$,  $|E|$, $|F|$,   $|\Phi|$, $|u^{i}|$, $|\rho|$, $|e|$,
$|R_{ij}|$, $|R|$, $(det g)^{-1}$, $|\frac{1}{\rho}|$, is uniformly
bounded over every bounded interval
$[0, T^{*}[$, where $T^{*}< +\infty$.\\
Notice that the hypothesis of Theorem~\ref{ttheo1} are satisfied; so
(\ref{2.15}) applies, i.e H is bounded.
\begin{enumerate}
    \item[$\bullet$]We deduce at once from (\ref{2.21})
    and using (\ref{2.15}) that $\Phi$ is bounded.
    \item[$\bullet$] Since the l.h.s of (\ref{2.20}) is bounded
    and since
     $-R>0$, $\tau_{00}\geq0$, $\sigma^{ij}\sigma_{ij}\geq
    0$, $U\geq0$, $\rho \geq 0$; we deduce that U and $\rho u^{2}_{0}=\rho (u_{0})^2$ are
    bounded;  but $\rho (u_{0})^2\geq \rho \geq 0$, since
    $u^{0}\geq1$; then we deduce that $\rho$ is bounded.
    \item [$\bullet$](\ref{2.18}) gives, using (\ref{2.17})
    to express $\sigma^{ij}\sigma_{ij}$,  and since: $U\geq 0$, $ g^{ij}u_{i}u_{j}\geq
    0$, $\rho u^{2}_{0}\geq0$, $\tau_{00}\geq0$:
\begin{equation}\label{3.18}
    \frac{d H}{dt}\geq k_{ij}k^{ij}-\Lambda-4 \pi m^2 \Phi^2.
\end{equation}
Then integrating (\ref{3.18}) over $[0, t]$, $0<t \leq T^{*}$, we
have, since H and $\Phi$ are bounded:
\begin{equation}\label{3.19}
    \int_{0}^{T^{*}}k_{ij}k^{ij}(s)ds<+\infty
\end{equation}
Next we have, integrating (\ref{1.46}) over $[0, t]$, $t\in [0,
T^{*}]$:
\begin{equation}\label{3.20}
    |g(t)|\leq |g^{0}|+2\int_{0}^{t}|k(s)|ds
\end{equation}
but setting $A_{1}=(g^{ij})$, $A_{2}=(k_{ij})$, (\ref{2.28}) gives:
\begin{equation}\label{3.21}
    \|k\|\leq\|g\|(k_{ij}k^{ij})^{\frac{1}{2}}
\end{equation}
and we deduce from (\ref{3.20}) and (\ref{3.21}) that:
\begin{equation*}
    \|g(t)\|\leq
    \|g^{0}\|+2\int_{0}^{t}\|g(s)\|(k_{ij}k^{ij})^{\frac{1}{2}}(s)ds.
\end{equation*}
Hence, by Gronwall Lemma, there exists a constant $C>0$ such that:
\begin{equation}\label{3.22}
    \|g(t)\|\leq C
    \|g^{0}\|\exp(C\int_{0}^{t}(k_{ij}k^{ij})^{\frac{1}{2}}(s))dt
\end{equation}
but we deduce from (\ref{3.22}) applying Schwarz inequality, using
(\ref{3.19}) and since $T^{*}< +\infty$, that $\|g\|$ and hence
$|g|$ is bounded.

    \item [$\bullet$]By (\ref{1.55}), $R_{ij}$ expresses in
    terms of $\gamma^{k}_{ij}$ given itself by
    (\ref{1.19}), which involves $g^{ij}$; so we need to
    control $(det g)^{-1}$. We use once more the formula: $$\frac{d}{dt}[\ln(det g)]=g^{ij}\frac{d
    g^{ij}}{dt}.$$
    Then, using the evolution equation (\ref{1.46}), we obtain:
    \begin{equation}\label{3.23}
\frac{d}{dt}[\ln(det g)]=-2H.
    \end{equation}
Since H is bounded, we obtain, by integrating (\ref{3.23}):
\begin{equation*}
    -C\leq \ln (det g)\leq C
\end{equation*}
where $C>0$ is a constant. Hence:
\begin{equation*}
 e^{-C}\leq det g\leq e^{C}
\end{equation*}
which shows that, both $detg$ and $(det g)^{-1}$ are bounded. Then,
by (\ref{1.55}), $R_{ij}$ is bounded and $R=g^{ij}R_{ij}$ is
bounded.

    \item [$\bullet$]
    Now consider the Hamiltonian constraint (\ref{1.56})
    which gives, since $\tau_{00}\geq0$, $\rho u^{2}_{0}\geq0$, $T_{00}\geq 0$ (see(\ref{1.23}))
\begin{equation*}
    R+H^2-2\Lambda \geq k_{ij}k^{ij}.
\end{equation*}
Then, since R and $H^2$ are bounded, $k_{ij}k^{ij}$ which is
positive is bounded. Then, since $\|g\|$ is bounded, by (\ref{3.21})
, $\|k\|$ and hence $|k|$, is bounded.

    \item [$\bullet$]Deduce from (\ref{2.20}) whose l.h.s is
    bounded and using  $U\geq0$, $\tau_{00}\geq0$, $\rho u^{2}_{0}\geq0$,  $\sigma^{ij}\sigma_{ij}\geq
    0$,   $-R>0$, that $\tau_{00}$ is bounded. But by
    (\ref{1.25}):

    \begin{equation}\label{3.24}
        \tau_{00}=\frac{1}{2}g_{ij}E^{i}E^{j}+\frac{1}{4}F^{ij}F_{ij}
    \end{equation}
which implies:
\begin{equation}\label{3.25}
0\leq \frac{1}{2}g_{ij}E^{i}E^{j}\leq \tau_{00} \;\; ; \;\; 0\leq
\frac{1}{4}F^{ij}F_{ij}\leq \tau_{00}
\end{equation}
and $g_{ij}E^{i}E^{j}$ is bounded. Hence, using (\ref{2.29}) and
since $|g|$ is bounded, $E^{i}$ is bounded.
    \item [$\bullet$]The constraint (\ref{1.59}) shows,
    since $u^{0}\geq 1$ and $E^{i}$ is bounded, that $e$ is bounded.
    \item [$\bullet$]Integrating the equation (\ref{1.49}) in
    $F_{ij}$ over $[0, t]$, $t\leq T^{*}<+\infty $ shows, since
    $|g|$ and $|E|$ are bounded, that $|F|$ is bounded.
    \item [$\bullet$] It remains the cases of $\frac{1}{\rho}$ and
    $u^{i}$.\\
    Expression (\ref{1.32}) of $\rho$ gives, using the
    notations (\ref{1.45}):
\begin{equation}\label{3.26}
\Big(\frac{1}{\rho}\Big)(t)=\frac{u^0}{\rho^{0}u^0(0)}\exp\Big[\int_{0}^{t}(-H+C^{i}_{ij}\frac{u^{j}}{u^{0}})(s)ds\Big].
\end{equation}
Now as we already indicated, equation (\ref{1.8}) gives for
$\beta=0$
\begin{equation}\label{3.27}
\dot{u}^0=k_{ij}\frac{u^{i}u^{j}}{u^0}-\Big(4\pi
eg_{ij}\frac{u^{i}E^{j}}{u^0}\Big)\frac{1}{\rho}.
\end{equation}
We deduce from (\ref{3.26}) and (\ref{3.27}) that:
\begin{equation}\label{3.28}
\dot{u}^0\leq G(t)u^{0}
\end{equation}
where:
\begin{equation}\label{3.29}
 G(t)=\frac{|k_{ij}u^{i}u^{j}|}{(u^{0})^2}+\frac{|4\pi
eg_{ij}E^{j}|}{\rho^{0}u^{0}(0)}.\frac{|u^{i}|}{u^{0}}
\exp\Big[\int_{0}^{t}(-H+C^{i}_{ij}\frac{u^{j}}{u^{0}})(s)ds\Big].
\end{equation}
Integrating (\ref{3.28}) over $[0. t]$, $t\leq T^{*}< +\infty$
gives:
\begin{equation}\label{3.30}
    u^{0}(t)\leq u^{0}(0)\exp\Big[\int_{0}^{t}G(s)ds\Big]
\end{equation}
then, using (\ref{2.29}) to bound $\dfrac{u^{i}}{u^{0}}$ and since
$k$, $e$, $g$, $E$, $H$ are bounded, by (\ref{3.30}), there exists a
constant $C(T^{*})>0$ such that $u^{0}\leq C(T^{*})u^{0}(0)$. Hence
by (\ref{3.26}), $\dfrac{1}{\rho}$ is bounded and writing
$u^{i}=\frac{u^{i}}{u^{0}}u^{0}$ shows that $u^{i}$ is bounded. This
completes the proof of theorem~\ref{ttheo33}
\end{enumerate}
\end{proof}
\section{Asymptotic behaviour}
\label{para4}
 We consider the global solution over $[0, +\infty[$
and we investigate the asymptotic behaviour of the different
elements at late times. We introduce the following quantity which
plays a key role:
\begin{equation}\label{4.1}
    Q=H^2-24 \pi T_{00}-3\Lambda
\end{equation}
At late times, we have the following asymptotic behaviour:
\begin{theorem}\label{ttheo4.1}
\begin{eqnarray}
   Q&=&\mathcal{O}(e^{-2\gamma t}) \label{4.2} \\
   \sigma_{ij}\sigma^{ij}&=& \mathcal{O}(e^{-2\gamma t}) \label{4.3}\\
   |R|&=& \mathcal{O}(e^{-2\gamma t}) \label{4.4}\\
   \tau_{00} &=&\mathcal{O}(e^{-2\gamma t})  \label{4.5}\\
   \rho &=&\mathcal{O}(e^{-2\gamma t})  \label{4.6}\\
   (\dot{\Phi})^{2}&=& \mathcal{O}(e^{-2\gamma t}) \label{4.7}\\
   F^{ij}F_{ij}&=&\mathcal{O}(e^{-2\gamma t})  \label{4.8}\\
  E^{i}E_{i} &=&\mathcal{O}(e^{-2\gamma t}) \label{4.9}\\
   \|\sigma(t)\|&\leq& C e^{-\gamma t} \|g(t)\| \label{4.10}\\
   \Phi^2 &\longrightarrow &  L>0  \label{4.11}\\
   T_{00} &\longrightarrow & \frac{m^2}{2} L  \label{4.12}\\
    H &=& -(3C_{0})^{\frac{1}{2}}+ \mathcal{O}(e^{-2 \delta t}) \label{4.13}\\
   |e^{-2\delta t}g_{ij}| &\leq& C   \label{4.14}\\
    |e^{2\delta t}g^{ij}|  &\leq& C \label{4.15}\\
   g_{ij}(t) &=&  e^{2\delta t}(G_{ij}+\mathcal{O}(e^{-\gamma t}))  \label{4.16}\\
  g^{ij}(t) &= & e^{-2 \delta t}(G^{ij}+\mathcal{O}(e^{-\gamma t}))  \label{4.17}\\
   |k_{ij}| &\leq & C e^{2 \delta t}  \label{4.18}\\
   |E^{i}| &\leq & C e^{\nu t} \label{4.19}\\
   |F_{ij}| &\leq & C e^{(2\delta+\nu)t} \label{4.20}\\
   |F^{ij}| &\leq &Ce^{(\delta-\gamma)t} \label{4.21}\\
   |e| &\leq & Ce^{\nu t}  \label{4.22}
\end{eqnarray}
where
\begin{eqnarray}
\nonumber&&\gamma=\Big[\frac{\Lambda+4 \pi m^2
(\Phi^{0})^2}{3}\Big]^{\frac{1}{2}}; \;\;C_0=\Lambda+4\pi m^2
L;\;\;\delta=[\frac{1}{3}(\Lambda+4\pi m^2 L)]^{\frac{1}{2}} ;  \\
         && \nu = 3\delta - \gamma;  \label{4.23}
\end{eqnarray}
$G_{ij}$ a symmetric positive definite constant matrix;
$(G^{ij})=(G_{ij})^{-1}$.
\end{theorem}
\begin{proof}
\hspace*{2cm}\\
 Notice that by (\ref{1.23}) and (\ref{1.34}),
we have: $m^2\Phi^2=2T_{00}-2U$. Expression (\ref{4.1}) of Q then
shows, using (\ref{2.20}) that $Q\geq 0$. Let us point out first of
all that, the quantity Q defined by (\ref{4.1}) plays a key role in
GR, and in the presence of the massive scalar field, it stands for
the quantities S in \cite{13}, Z in \cite{20}, $\widetilde{S}$ in
\cite{11}, $\overline{S}$ in \cite{7} and reduces to $H^2\pm 3
\Lambda$ in \cite{26} which deals with the case of zero scalar
field.
\begin{enumerate}
    \item [$\bullet$]We have, using the expression (\ref{1.23}) of
    $T_{00}$ and (\ref{1.34}):
\begin{equation}\label{4.24}
T_{00}= U +\frac{1}{2}m^2\Phi^2
\end{equation}
then, the evolution equations (\ref{1.51}) and (\ref{1.52}) in
$\Phi$ and U give:
\begin{equation}\label{4.25}
    \dot{T}_{00}=2HU.
\end{equation}
Expression (\ref{4.1}) of Q then gives, using (\ref{4.25}):
$$\dot{Q}=2H(\dot{H}-24\pi U)$$ then, using equation (\ref{2.1}) in H,
in which we use the Hamiltonian constraint (\ref{1.56}) to express
$R+H^2$ and (\ref{2.17}) to express $k_{ij}k^{ij}$ and since by
(\ref{1.23}) and (\ref{1.34}): $4\pi g^{ij} T_{ij}=24 \pi U-
12T_{00}$,
 we obtain
 \begin{equation}\label{4.26}
\dot{Q} =\frac{2}{3}H\Big[H^2-24\pi
T_{00}-3\Lambda+3(8\pi\tau_{00}+4\pi\rho u^2_0+4\pi
g^{ij}u_{i}u_{j}+\sigma_{ij}\sigma^{ij})\Big].
 \end{equation}
 But since $8\pi\tau_{00}+4\pi\rho u^2_0+4\pi
g^{ij}u_{i}u_{j}+\sigma_{ij}\sigma^{ij}\geq0$, $H<0$, and given the
definition (\ref{4.1}) of Q, (\ref{4.26}) gives:
\begin{equation}\label{4.27}
\dot{Q}\leq \frac{2}{3}HQ.
\end{equation}
Integrating (\ref{4.27}) over $[0, t]$, $t>0$ yields: $$0\leq Q \leq
Q_{0}\exp [\int_{0}^{t}\frac{2}{3}Hds];$$ and (\ref{2.15}) gives
$$0\leq Q \leq Q_{0}\exp\Big(-2t\sqrt{\frac{1}{3}(\Lambda+4 \pi m^2
(\Phi^{0})^2)}\Big)$$ and (\ref{4.2}) follows.
    \item [$\bullet$]Using (\ref{2.20}), the results (\ref{4.3}),
    (\ref{4.4}), (\ref{4.5}), (\ref{4.6}) and (\ref{4.7}) (since
    $(\dot{\Phi})^2=2U$) are direct consequences of (\ref{4.2}).
    \item [$\bullet$](\ref{4.8}) and (\ref{4.9}) are consequences of
    (\ref{4.5}), using (\ref{3.25}).
    \item [$\bullet$]Setting $A_{1}=(g^{ij})$ and
    $A_{2}=(\sigma_{ij})$, (\ref{2.28}) gives:
    \begin{equation}\label{4.28}
        \|\sigma(t)\|\leq \|g(t)\|(\sigma_{ij}\sigma^{ij})^{\frac{1}{2}}
    \end{equation}
(\ref{4.10}) then follows from (\ref{4.3}).
    \item [$\bullet$]The evolution equations (\ref{1.51}) and
    (\ref{1.52}) in $\Phi$ and U give:
    \begin{equation}\label{4.29}
m^2\Phi\dot{\Phi}+\dot{U}=\frac{d}{dt}\Big[\frac{m^2}{2}\Phi^2+U\Big]=2HU.
    \end{equation}
But since $H<0$ and $U>0$, (\ref{4.29}) implies:
$\frac{d}{dt}[\frac{m^2}{2}\Phi^2+U]\leq 0$; we then deduce, using
$U>0$ that:
\begin{equation}\label{4.30}
\frac{m^2}{2}\Phi^2\leq \frac{m^2}{2}\Phi^2+U\leq
\frac{m^2}{2}(\Phi^{0})^2+U^{0}.
\end{equation}
Hence $\Phi^{2}$ is bounded. But the evolution equation (\ref{1.51})
in $\Phi$ shows that $\dot{\Phi}>0$; then $\Phi\geq \Phi^{0}>0$ and
since $\frac{d}{dt}(\Phi^{2})=2\Phi \dot{\Phi}>0$, $\Phi^{2}$ is an
increasing function. $\Phi^{2}$ being positive, increasing and
bounded has a strictly positive limit, i.e, there exits $L>0$ such
that
\begin{equation}\label{4.31}
\Phi^{2}\longrightarrow L
\end{equation}
with:
\begin{equation}\label{4.32}
\Phi^{2}\leq L
\end{equation}
and we have (\ref{4.11}).
    \item [$\bullet$](\ref{4.12}) follows from
    (\ref{4.24}), (\ref{4.7}) and (\ref{4.11}).
    \item [$\bullet$]To prove (\ref{4.13}) which is one for the main
    results, since by (\ref{4.2}) $Q \longrightarrow 0$ its
    expression (\ref{4.1}) shows, using (\ref{4.12}) that:
    \begin{equation}\label{4.33}
H^2-3\Lambda \longrightarrow 12 \pi m^2 L.
    \end{equation}
Hence:
$$H^2- (3\Lambda+12\pi m^2 L)=\Big[H-(3\Lambda+12\pi m^2 L)^{\frac{1}{2}}\Big]
 \Big[H+(3\Lambda+12\pi m^2 L)^{\frac{1}{2}}\Big]\longrightarrow 0. $$
But by (\ref{2.15}), $H<0$; so:
$$H-(3\Lambda+12\pi m^2L)^{\frac{1}{2}}<-(3\Lambda+12\pi m^2 L)^{\frac{1}{2}}< 0. $$
We then deduce that:
\begin{equation}\label{4.34}
       H \longrightarrow -(3C_0)^{\frac{1}{2}}
\end{equation}
where:
\begin{equation}\label{4.35}
C_0=\Lambda + 4 \pi m^2L>0.
\end{equation}
But since by (\ref{2.22}) H is an increasing function (\ref{4.34})
implies:
\begin{equation}\label{4.36}
      H \leq -(3C_0)^{\frac{1}{2}}.
\end{equation}
Now (\ref{4.32}): $-4\pi m^2 \Phi^2\geq -4\pi m^2 L$; (\ref{2.19})
then implies:
\begin{equation}\label{4.37}
    \frac{dH}{dt}\geq \frac{H^2}{3}-C_{0}
\end{equation}
write:
$$\frac{H^{2}}{3}-C_0=\frac{1}{3}(H^2-3C_0)=\frac{1}{3}[-H+(3C_0)^{\frac{1}{2}}][-H-(3C_0)^{\frac{1}{2}}]$$
in which using (\ref{4.36}) we have: $$-H+(3C_0)^{\frac{1}{2}}\geq
(3C_{0})^{\frac{1}{2}}+(3C_{0})^{\frac{1}{2}}=2(3C_{0})^{\frac{1}{2}};
$$
(\ref{4.36}) also implies: $-H-(3C_0)^{\frac{1}{2}}\geq 0$. We then
deduce from (\ref{4.37})
\begin{equation}\label{4.38}
    \frac{d H}{dt}\geq 2 \delta[-H-(3C_0)^{\frac{1}{2}}]
\end{equation}
where:
\begin{equation}\label{4.39}
    \delta=\frac{1}{3}(3C_0)^{\frac{1}{2}}.
\end{equation}
Write (\ref{4.38}) in the form:
\begin{equation}\label{4.40}
    \frac{d}{dt}(H+(3C_0)^{\frac{1}{2}})+ 2
\delta[H+(3C_0)^{\frac{1}{2}}] \geq 0.
\end{equation}
Multiply (\ref{4.40}) by $e^{2\delta t}>0$ and integrate over $[0,
t]$ to obtain:
\begin{equation}\label{4.41}
    e^{2\delta t}[H+(3C_0)^{\frac{1}{2}}] \geq
    H(0)+(3C_0)^{\frac{1}{2}}.
\end{equation}
Now multiply (\ref{4.41}) by $-e^{-2\delta t}<0$, use once more
(\ref{4.36}) which gives $H+(3C_0)^{\frac{1}{2}}< 0$ to obtain:
\begin{equation*}
    |H+(3C_0)^{\frac{1}{2}}|\leq |H(0)+(3C_0)^{\frac{1}{2}}|e^{-2\delta
    t}
\end{equation*}
and (\ref{4.13}) follows with , see (\ref{4.35}) and (\ref{4.39}):
$\delta=[\frac{1}{3}(\Lambda+4\pi m^2 L)]^{\frac{1}{2}}$.
    \item [$\bullet$]To prove (\ref{4.14}), set $h_{ij}=e^{-2\delta
    t}g_{ij}$; then we have, using equation (\ref{1.46}) in $g_{ij}$:
\begin{equation}\label{4.42}
    \frac{d h_{ij}}{dt} = -2\delta h_{ij}-2k_{ij}e^{-2\delta t}.
\end{equation}
Now, using (\ref{2.16}) to express $k_{ij}$, the result
(\ref{4.13}), the expression of $C_{0}$ and $\delta$ in (\ref{4.23})
we deduce from (\ref{4.42}) that:
\begin{equation}\label{4.43}
\frac{d h_{ij}}{dt}  = \mathcal{O}(e^{-2\delta
t})h_{ij}-2e^{-2\delta t}\sigma_{ij}.
\end{equation}
Integrating (\ref{4.43}) over $[0, t]$, $t>0$, and taking the norm
yields:
\begin{equation}\label{4.44}
    \|h(t)\| \leq\|h(0)\|+C\int_{0}^{t}(e^{-2\delta s}\|h(s)\|+e^{-2\delta
  s}\|\sigma(s)\|)ds.
\end{equation}
where $C>0$ is a constant. Notice that $h^{ij}=e^{2\delta t}g^{ij}$.
Now setting in formula (\ref{2.28}): $A_{1}=(h^{ij})$,
$A_{2}=(\sigma_{ij})$ yields:
\begin{equation*}
    \|\sigma(s)\| \leq \|h(s)\|(\sigma_{h}^{ij}\sigma_{ij})^{\frac{1}{2}}
\end{equation*}
where $\sigma_{h}^{ij}=h^{il}h^{jk}\sigma_{lk}=e^{4\delta
t}\sigma^{ij}$. Hence, $ \|\sigma(s)\| \leq \|h(s)\|e^{2\delta
t}(\sigma^{ij}\sigma_{ij})^{\frac{1}{2}}$; (\ref{4.3}) then gives
$$ \|\sigma(s)\| \leq \|h(s)\|e^{2\delta s}e^{-\gamma s}.$$
We then deduce from (\ref{4.44}):
\begin{equation*}
    \|h(t)\| \leq\|h(0)\|+C \int_{0}^{t}(e^{-2\delta s}\|h(s)\|+e^{-\gamma s}\|h(s)\|)ds
\end{equation*}
then, since $\gamma < \delta$ [see (\ref{4.23})]:
\begin{equation*}
    \|h(t)\| \leq\|h(0)\|+C\int_{0}^{t}e^{-\gamma s}\|h(s)\|ds.
\end{equation*}
By Gronwall Lemma, this gives:
\begin{equation*}
     \|h(t)\| \leq \|h(0)\|\exp[\int_{0}^{t}Ce^{-\gamma s}]\leq C,
\end{equation*}
where $C>0$ is a constant, and (\ref{4.14}) follows.
    \item [$\bullet$]To obtain (\ref{4.15}) set this time $L^{ij}=e^{2\delta
    t}g^{ij}$ and proceed as for (\ref{4.13}) and obtain
    (\ref{4.15}).
\item [$\bullet$] To prove (\ref{4.16}) which is one of the main
results, first use
\begin{equation}\label{4.45}
    e^{-2\delta t}\sigma_{ij}=\mathcal{O}(e^{-\gamma t})
\end{equation}
which is a direct consequence of (\ref{4.28}), (\ref{4.3}) and
(\ref{4.14}). Recall that setting $h_{ij}= e^{-2\delta t}g_{ij}$ led
to (\ref{4.43}). We deduce from (\ref{4.43}), using (\ref{4.45}),
$\gamma < \delta$, and since $h$ is bounded:
\begin{equation}\label{4.46}
     \frac{d h_{ij}}{dt}=\mathcal{O}(e^{-\gamma t}).
\end{equation}
(\ref{4.46}) shows that $\dot{h}_{ij}$ has an exponential fall of at
late times and by the mean value theorem, $h_{ij}$ has a limit we
denote $G_{ij}$ as $t\longrightarrow +\infty$. Then we can write:
\begin{equation}\label{4.47}
    h_{ij}(t)=G_{ij}+\mathcal{O}(e^{-\gamma t})
\end{equation}
where, given the properties of $h_{ij}$, $G_{ij}$ is a symmetric,
positive definite constant $3\times3$ matrix. (\ref{4.16}) follows
from (\ref{4.47}) since $g_{ij}=e^{2\delta t}h_{ij}$.
    \item [$\bullet$] (\ref{4.17}) is a direct consequence of
    (\ref{4.16}) since $(g^{ij})=(g_{ij})^{-1}$.
    \item [$\bullet$]To prove (\ref{4.18}), use (\ref{2.16}) which
    implies, since H is bounded
    \begin{equation}\label{4.48}
        \|k\| \leq C ( \|g\|+\|\sigma\|).
    \end{equation}
(\ref{4.18}) then follows from (\ref{4.48}), using (\ref{4.16}) and
(\ref{4.45}).
    \item [$\bullet$]To obtain (\ref{4.19}), use (\ref{2.29}) which
    gives $|E^{i}|\leq (E^{j}E_{j})^{\frac{1}{2}}|g|^{\frac{3}{2}}$
    and conclude by applying (\ref{4.9}) and (\ref{4.16}).
    \item [$\bullet$]To obtain (\ref{4.20}), integrate equation
    (\ref{1.49}) in $F_{ij}$ and use (\ref{4.16}) and (\ref{4.19}).
    \item [$\bullet$]To obtain (\ref{4.21}), use
    $F^{ij}=g^{il}g^{jk}F_{kl}$, (\ref{4.17}) and (\ref{4.20}).
    \item [$\bullet$]Finally to obtain (\ref{4.22}), use the
    constraint equation (\ref{1.59}), (\ref{4.19}) and\\
     $u^{0}\geq
    1$. This completes the proof of Theorem~\ref{ttheo4.1}.
\end{enumerate}
\end{proof}
\begin{remark}\label{rq4.2}
\hspace*{2cm}\\
The result (\ref{4.16}) shows an exponential growth of the metric
tensor $g$ at late times. this result, but also (\ref{4.13}) confirm
mathematically the accelerated expansion of the universe as observed
in Astrophysics.
\end{remark}
\section{Geodesic Completeness}
\label{para5} We prove:
\begin{theorem}\label{ttheo5.1}
\hspace*{2cm}\\
 The space-time which exists globally is future
geodesically complete.
\end{theorem}
\begin{proof}
\hspace*{2cm}\\
We use the fact that, the geodesics equations for the metric
(\ref{1.1}) imply that, along the geodesics whose affine parameter
is denoted by $s$, the variables $t$, $u^0$, $u^{i}$ satisfy a first
order differential system containing between others, the equation:
\begin{equation}\label{5.1}
\frac{dt}{ds}=u^{0}.
\end{equation}
The space-time will be future geodesically  complete if we prove
that the affine parameter $s$ also goes to infinity. Then using
(\ref{1.11}), the notations (\ref{1.45}) and (\ref{5.1}), it will
enough if we could prove that:
\begin{equation}\label{5.2}
\frac{d s}{dt}=(1+g_{ij}u^{i}u^{j})^{-\frac{1}{2}}\geq C > 0
\end{equation}
where $C>0$ is a constant, since one could then deduce at once from
(\ref{5.2}), integrating, that $s\geq Ct+D$ where D is a constant;
hence $s\longrightarrow +\infty$ as $t\longrightarrow +\infty$. Our
goal will then be to prove, that $(1+g_{ij}u^{i}u^{j})$ or finally
$g_{ij}u^{i}u^{j}$, is bounded by a strictly positive constant. For
this purpose we use equation (\ref{1.50}) in $u^{i}$. It shows to be
useful considering $u_{i}=g_{ij}u^{i}$, rather than $u^{i}$.
Differentiating this relation, we have:
\begin{equation}\label{5.3}
\frac{d u_{i}}{dt}=g_{ij}\frac{d u^{j}}{dt}+u^{j}\frac{d g_{ij}}{dt}
\end{equation}
then, using equation (\ref{1.50}) in $u^{i}$ and equation
(\ref{1.46}) in $g_{ij}$, we deduce from (\ref{5.3}), the equation:
\begin{equation}\label{5.4}
\frac{du^{i}}{dt}
=-g_{ij}\gamma^{j}_{lk}\frac{u^{l}u^{k}}{u^{0}}+g_{ij}(-4\pi\frac{e}{\rho}E^{j}
+4\pi\frac{e}{\rho}F^{\;\;j}_{l}\frac{u^{l}}{u^{0}}).
\end{equation}
In order to bound $(g_{ij}u^{i}u^{j})$ we set up a differential
equation for this quantity. First notice that equation (\ref{1.46})
gives, using $g^{ij}g_{il}=\delta^{j}_{l}:$ $$\frac{d
g^{ij}}{dt}=2k^{ij};$$ from where we deduce:
\begin{equation}\label{5.5}
\frac{d}{dt}(g^{ij}u_{i}u_{j})=2k^{ij}u_{i}u_{j}+2g^{ij}u_{j}\frac{d
u_{i}}{dt}.
\end{equation}
A direct calculation using (\ref{5.4}) shows that, in (\ref{5.5}) we
have:
\begin{equation}\label{5.6}
2g^{ij}u_{j}\frac{du_{i}}{dt}=-8\pi\frac{e}{\rho}u_{i}E^{i}+8\pi\frac{e}{\rho}F_{lk}\frac{u^{l}u^{k}}{u^{0}}
 -2\gamma_{jk}^{i}\frac{u^{j}u^{k}u_{i}}{u^{0}}.
\end{equation}
But given the antisymmetry of $F_{lk}$, we have
$F_{lk}u^{l}u^{k}=0$. Next using the expression (\ref{1.19}) of
$\gamma_{ij}^{l}$ we obtain for the last term in (\ref{5.6}):
\begin{equation}\label{5.7}
\gamma_{jk}^{i}u^{j}u^{k}u_{i} =
-\frac{1}{2}g_{jl}C^{l}_{kp}u^{j}u^{k}u^{p}
  +\frac{1}{2}g_{kl}C^{l}_{pj}u^{j}u^{k}u^{p}
   +\frac{1}{2}C^{i}_{jk}u^{j}u^{k}u_{i}.
\end{equation}
But since $C^{l}_{jk}=-C^{l}_{kj}$ the last term in (\ref{5.7})
vanishes. For the same reason:
\begin{equation*}
-\frac{1}{2}g_{jl}C^{l}_{kp}u^{j}u^{k}u^{p}+\frac{1}{2}g_{kl}C^{l}_{pj}u^{j}u^{k}u^{p}
=\frac{1}{2}(C^{l}_{pk}+C^{l}_{kp})g_{jl}u^{j}u^{k}u^{p}=0.
\end{equation*}
Consequently, the r.h.s of (\ref{5.6}) reduces to its first term and
(\ref{5.5}) gives:
\begin{equation}\label{5.8}
    \frac{d}{dt}(g^{ij}u_{i}u_{j})=2k^{ij}u_{i}u_{j}-8\pi\frac{e}{\rho}u_{i}E^{i}.
\end{equation}
But by (\ref{2.16}), $k^{ij}=\frac{H}{3}g^{ij}+\sigma^{ij}$;
(\ref{5.8}) then gives, using (\ref{4.13}) and (\ref{4.23}):
\begin{equation}\label{5.9}
\frac{d}{dt}(g^{ij}u_{i}u_{j})=\Big[-2\delta+\mathcal{O}(e^{-2\delta
t})\Big]g^{ij}u_{i}u_{j}+2\sigma^{ij}u_{i}u_{j}-8\pi\frac{e}{\rho}u_{i}E^{i}.
\end{equation}
Now we have $\sigma^{ij}=g^{ik}g^{jl}\sigma_{kl}$, so by
(\ref{4.45}) and (\ref{4.17}): $e^{(2\delta+\gamma)t}\sigma^{ij}$ is
bounded. This implies, since the matrix $G^{ij}$ is positive
definite and constant, that there exits a constant $C>0$ such that:
\begin{equation}\label{5.10}
\sigma^{ij}u_{i}u_{j}\leq Ce^{-(2\delta + \gamma)t}
G^{ij}u_{i}u_{j}.
\end{equation}
Now by (\ref{4.17}), there exists a constant $C>0$ such that:
\begin{equation}\label{5.11}
G^{ij}u_{i}u_{j}\leq Ce^{2\delta t}g^{ij}u_{i}u_{j}.
\end{equation}
Now deduce from (\ref{1.32}) and (\ref{1.36bis}) that:
\begin{equation}\label{5.12}
    \frac{e}{\rho}=\frac{e^{0}}{\rho^{0}},
\end{equation}
we then obtain from (\ref{5.9}), using (\ref{5.10}), (\ref{5.11})
and (\ref{5.12}):
\begin{equation}\label{5.13}
\frac{d}{dt}(g^{ij}u_{i}u_{j})\leq(-2\delta+Ce^{-2\delta t})g^{ij}u_{i}u_{j}
+Ce^{-\gamma t}g^{ij}u_{i}u_{j}+C|u_{i}E^{i}|. \\
\end{equation}
Now set:
\begin{equation}\label{5.14}
    W=e^{2\delta t}g^{ij}u_{i}u_{j}
\end{equation}
then we have, using (\ref{5.13}):
\begin{eqnarray}
\nonumber\frac{d W}{dt}&=&2\delta e^{2\delta t}g^{ij}u_{i}u_{j}+e^{2\delta t}\frac{d}{dt}(g^{ij}u_{i}u_{j}) \\
           &\leq &Cg^{ij}u_{i}u_{j}+Ce^{(2\delta-\gamma) t}g^{ij}u_{i}u_{j}+Ce^{2\delta t} |u_{i}E^{i}|.\label{5.15}
\end{eqnarray}
\noindent Now since $g$ is a scalar product: $|u_{i}E^{i}|\leq
(u^{i}u_{i})^{\frac{1}{2}}(E^{i}E_{i})^{\frac{1}{2}}$. so, by
(\ref{4.9}) we have $|u_{i}E^{i}|\leq
C(g^{ij}u_{i}u_{i})^{\frac{1}{2}}e^{-\gamma t} $; (\ref{5.15}) then
gives:
\begin{equation}\label{5.16}
\frac{d W}{dt}\leq C e^{-2 \delta t}W +C e^{-\gamma t}W+ C
e^{(\delta-\gamma) t}W^{\frac{1}{2}}
\end{equation}
(\ref{5.16}) gives, since $0< \gamma < \delta$:
\begin{equation}\label{5.17}
\frac{d W}{dt}\leq C e^{- \gamma t}W + C e^{\delta
    t}W^{\frac{1}{2}}.
\end{equation}
But it is well known that by (\ref{5.17}) we have:
\begin{equation}\label{5.18}
    W(t)\leq z(t)
\end{equation}
where
\begin{equation}\label{5.19}
   \large \left\{
      \begin{array}{ll}
\frac{d z}{dt}= C e^{- \gamma t}z + C e^{\delta t}z^{\frac{1}{2}} & \hbox{} \\
z(0)=W(0).  & \hbox{}
      \end{array}
    \right.
\end{equation}
But (\ref{5.19}) is a Bernoulli equation whose solution is:
\begin{equation}\label{5.20}
z(t)=\exp
\Big(-2\int_{0}^{t}a(s)ds\Big)\Big[(W(0))^{\frac{1}{2}}+\int_{0}^{t}b(s)(\exp\int_{0}^{s}a(\tau)d\tau)
ds\Big]^{2}
\end{equation}
where:
\begin{equation}\label{5.21}
    a(t)=-\frac{C}{2}e^{-\gamma t}; \;\;  b(t)=\frac{C}{2}e^{\delta
    t}.
\end{equation}
One deduces easily from (\ref{5.20}), (\ref{5.21}) that:
$$ z(t)\leq C e^{2\delta t},$$
where $C>0$ is a constant. Then using expression (\ref{5.14}) of W
and (\ref{5.18}), we obtain:
$$ g_{ij}u^{i}u^{j} \leq C .$$
This completes the proof of Theorem~\ref{ttheo5.1}.
\end{proof}
\section{Energy conditions}
\label{para6} In this section we prove that the global solution
satisfies the weak and the dominant energy conditions and, under
some hypothesis, the strong energy condition. Recall that a viable
physical theory is supposed to fulfill at least one of the energy
conditions (Hawking\cite{9}). In fact notice that considering the
stress-energy-matter tensor of the Einstein equations (\ref{1.2}),
and keeping the $(\widetilde{\;\;})$ to avo\"{\i}d any confusion,
the quantity
$$(\widetilde{T}_{\alpha\beta}+\widetilde{\tau}_{\alpha\beta}+
\rho
\widetilde{u}_{\alpha}\widetilde{u}_{\beta})\widetilde{V}^{\alpha}\widetilde{V}^{\beta}$$
represents physically, the energy density of the charged particle,
measured by an observer whose velocity is $\widetilde{V}^{\alpha}$
and so must be non-negative, $\widetilde{V}^{\alpha}$ being a future
pointing time-like vector, see \cite{27}.\\
We recall below the three types of energy conditions: let
$(\widetilde{V}^{\alpha})$ and $(\widetilde{W}^{\alpha})$ be any two
\textit{future pointing time-like vectors}. The solution is said to
satisfy:
\begin{enumerate}
    \item[1)]the weak energy condition if:
    \begin{equation}\label{6.1}
(\widetilde{T}_{\alpha\beta}+\widetilde{\tau}_{\alpha\beta}+ \rho
\widetilde{u}_{\alpha}\widetilde{u}_{\beta})\widetilde{V}^{\alpha}\widetilde{V}^{\beta}\geq
0.
    \end{equation}
    \item[2)]the strong energy condition if:
    \begin{equation}\label{6.2}
\widetilde{R}_{\alpha\beta}\widetilde{V}^{\alpha}\widetilde{V}^{\beta}\geq
0.
    \end{equation}
    \item[3)]the dominant energy condition if:
    \begin{equation}\label{6.3}
(\widetilde{T}_{\alpha\beta}+\widetilde{\tau}_{\alpha\beta}+ \rho
\widetilde{u}_{\alpha}\widetilde{u}_{\beta})\widetilde{V}^{\alpha}\widetilde{W}^{\beta}\geq
0.
    \end{equation}
\end{enumerate}
Obviously, (\ref{6.3}) implies (\ref{6.1}), just setting: $\widetilde{W}^{\alpha}=\widetilde{V}^{\alpha}$.\\
We begin by proving:

\begin{proposition}\label{prop6.1}
\hspace*{2cm}\\
Let $\widetilde{V}^{\alpha}$ and $\widetilde{W}^{\alpha}$ be two
future pointing time-like or null vectors. Then
\begin{equation}\label{6.4}
\widetilde{V}^{\alpha}\widetilde{W}_{\alpha}\leq 0
\end{equation}
\end{proposition}
\begin{proof}
\hspace*{2cm}\\
Since $\widetilde{V}^{\alpha}\widetilde{W}_{\alpha}=
\widetilde{g}_{\alpha\beta}\widetilde{V}^{\alpha}\widetilde{W}^{\beta}$
and given the definition (\ref{1.1}) of $\widetilde{g}$, (\ref{6.4})
is equivalent to:

\begin{equation}\label{6.5}
-\widetilde{V}^{0}\widetilde{W}^{0}+g_{ij}\widetilde{V}^{i}\widetilde{W}^{j}\leq0.
\end{equation}
Now, $\widetilde{V}^{\alpha}$, $\widetilde{W}_{\alpha}$ being future
pointing time-like or null we have:
\begin{equation*}
\widetilde{V}^{\alpha}\widetilde{V}_{\alpha}\leq0; \;\;
\widetilde{V}^{0}\geq0; \;\;
\widetilde{W}^{\alpha}\widetilde{W}_{\alpha}\leq0; \;\;
\widetilde{W}^{0}\geq0
\end{equation*}
,or, equivalently:
\begin{equation*}
0\leq g_{ij}\widetilde{V}^{i}\widetilde{V}^{j}\leq
(\widetilde{V}^{0})^{2};\; \widetilde{V}^{0}\geq0;\; 0\leq
g_{ij}\widetilde{W}^{i}\widetilde{W}^{j}\leq
(\widetilde{W}^{0})^{2};\;\widetilde{W}^{0}\geq0;
\end{equation*}
 hence:
\begin{numcases}
\strut
 \nonumber0\leq
(g_{ij}\widetilde{V}^{i}\widetilde{V}^{j})^{\frac{1}{2}}\leq
\widetilde{V}^{0}  \\
\nonumber0\leq
(g_{ij}\widetilde{W}^{i}\widetilde{W}^{j})^{\frac{1}{2}}\leq
\widetilde{W}^{0}
\end{numcases}
which gives:
\begin{equation}\label{6.6}
0\leq(g_{ij}\widetilde{V}^{i}\widetilde{V}^{j})^{\frac{1}{2}}(g_{ij}
\widetilde{W}^{i}\widetilde{W}^{j})^{\frac{1}{2}} \leq
\widetilde{V}^{0}\widetilde{W}^{0}.
\end{equation}
But since $g$ is a scalar product, we have:
\begin{equation}\label{6.7}
|g_{ij}\widetilde{V}^{i}\widetilde{W}^{j}|\leq
(g_{ij}\widetilde{V}^{i}\widetilde{V}^{j})^{\frac{1}{2}}(g_{ij}
\widetilde{W}^{i}\widetilde{W}^{j})^{\frac{1}{2}}
\end{equation}
(\ref{6.5}) then follows from (\ref{6.6}) and (\ref{6.7}).
\end{proof}
Next we prove this important result for the Maxwell tensor
$\widetilde{\tau}_{\alpha\beta}$ defined by (\ref{1.6}):
\begin{proposition}\label{prop6.2}
\hspace*{2cm}\\
For any two future pointing time-like vectors
$(\widetilde{V}^{\alpha})$, $(\widetilde{W}^{\alpha})$, we have:
\begin{equation}\label{6.8}
\widetilde{\tau}_{\alpha\beta}\widetilde{V}^{\alpha}\widetilde{W}^{\alpha}\geq
0.
\end{equation}
\end{proposition}
\begin{proof}
\hspace*{2cm}\\
It will be enough if we could prove (\ref{6.8}) by choosing any
suitable frame. Let us consider the frame of the four vectors:\\
$l=(l^{\alpha})$; $n=(n^{\alpha})$; $x=(x^{\alpha})$;
$y=(y^{\alpha})$, satisfying the following properties:
\begin{equation}\label{6.9}
    l_{\alpha}l^{\alpha}= n_{\alpha}n^{\alpha}= l_{\alpha}x^{\alpha}
    =l_{\alpha}y^{\alpha}=n_{\alpha}x^{\alpha}=
    n_{\alpha}y^{\alpha}=0.
\end{equation}
Now inspired for  instance by the case where the electromagnetic
fields $\widetilde{F}_{\alpha\beta}$ derives from a potential
vector, the antisymmetric 2-form $\widetilde{F}_{\alpha\beta}$ can
be written in one of the two following general forms:
\begin{equation}\label{6.10}
\widetilde{F}_{\alpha\beta}=\frac{A}{2}(l_{\alpha}n_{\beta}-l_{\beta}n_{\alpha})
+\frac{B}{2}(x_{\alpha}y_{\beta}-x_{\beta}y_{\alpha})
\end{equation}
or:
\begin{equation}\label{6.11}
\widetilde{F}_{\alpha\beta}=\frac{C}{2}(l_{\alpha}x_{\beta}-l_{\beta}x_{\alpha})
\end{equation}
where A, B, C are constants. In fact, since $M=\mathbb{R}\times G$,
with G is a \textit{simply connected} Lie group, by Poincare Lemma,
$\widetilde{F}_{\alpha\beta}$ is an exact form. Then, there exits a
potential vector $\widetilde{A}_{\alpha}$ over M such that:
$\widetilde{F}_{\alpha\beta}=\widetilde{\nabla}_{\alpha}\widetilde{A}_{\beta}-
\widetilde{\nabla}_{\beta}\widetilde{A}_{\alpha}$ then formula
(\ref{6.10}) generalizes the form given by the development of the
above expression of $\widetilde{F}_{\alpha\beta}$ by applying
(\ref{1.21}) in the case of the frame $(e_{\alpha})$ whereas
(\ref{6.11}) corresponds to the case of the natural frame
$(\partial_{\alpha})$.
\begin{enumerate}
    \item[$\bullet$]Next, it shows useful to choose the constants
    A, B, C by assuming that we have in addition:
    \begin{equation}\label{6.12}
        l_{\alpha}n^{\alpha}=-1; \;\;x_{\alpha}x^{\alpha}=y_{\alpha}y^{\alpha}=1;
        \;\;x_{\alpha}y^{\alpha}=0
    \end{equation}
    \item[$\bullet$]Now consider the Maxwell tensor (\ref{1.6}) i.e
\begin{equation}\label{6.13}
\widetilde{\tau}_{\alpha
\beta}=-\frac{1}{4}\widetilde{g}_{\alpha\beta}\widetilde{F}^{\lambda\mu}\widetilde{F}_{\lambda\mu}+
\widetilde{F}_{\alpha\lambda}\widetilde{F}^{\;\;\;\lambda}_{\beta}
\end{equation}
\underline{$1^{\circ})$Consider the form (\ref{6.10})}.\\
A direct calculation using (\ref{6.9}),  (\ref{6.12}) gives:
\begin{equation}\label{6.14}
\widetilde{F}^{\lambda\mu}\widetilde{F}_{\lambda\mu}=\frac{B^2-A^2}{2}
;\;\;\widetilde{F}_{\alpha\lambda}\widetilde{F}^{\;\;\;\lambda}_{\beta}=
\frac{A^2}{4}(l_{\alpha}n_{\beta}+n_{\alpha}l_{\beta})
+\frac{B^2}{4}(x_{\alpha}x_{\beta}+y_{\alpha}y_{\beta})
\end{equation}
so that, in this case, (\ref{6.13}) and (\ref{6.14}) give:
\begin{equation}\label{6.15}
\widetilde{\tau}_{\alpha
\beta}=\frac{1}{4}\Big[A^2(l_{\alpha}n_{\beta}+n_{\alpha}l_{\beta})
+B^2(x_{\alpha}x_{\beta}+y_{\alpha}y_{\beta})+\widetilde{g}_{\alpha\beta}\Big(\frac{A^2-B^2}{2}\Big)\Big].
\end{equation}
Now if we express the vectors
$\widetilde{V}=(\widetilde{V}^{\alpha})$,
$\widetilde{W}=(\widetilde{W}^{\alpha})$ in the frame $(l, n, x, y)$
by:
\begin{numcases}
\strut
\widetilde{V}^{\alpha}=Ml^{\alpha}+Nn^{\alpha}+Px^{\alpha}+Q y^{\alpha}\label{6.16} \\
\widetilde{W}^{\alpha}=M'l^{\alpha}+N'n^{\alpha}+P'x^{\alpha}+Q'y^{\alpha}
\label{6.17}
\end{numcases}
and if we set:
\begin{equation}\label{6.18}
    \widetilde{h}_{\alpha\beta}=-l_{\alpha}n_{\beta}-n_{\alpha}l_{\beta}+x_{\alpha}x_{\beta}+y_{\alpha}y_{\beta}
\end{equation}
then a direct calculation, using (\ref{6.9}) and (\ref{6.12}) gives:
\begin{equation*}
\widetilde{h}_{\alpha\beta}\widetilde{V}^{\alpha}\widetilde{W}^{\beta}
=\widetilde{g}_{\alpha\beta}\widetilde{V}^{\alpha}\widetilde{W}^{\beta}=(-NM'-MN'+PP'+QQ').
\end{equation*}
Hence $\widetilde{g}_{\alpha\beta}=\widetilde{h}_{\alpha\beta}$. So,
we can express $\widetilde{g}_{\alpha\beta}$ in (\ref{6.15}) by
(\ref{6.18}) and this gives:
\begin{equation}\label{6.19}
    \widetilde{\tau}_{\alpha\beta}=\frac{A^2+B^2}{8}\widetilde{\overline{\tau}}_{\alpha\beta}
\end{equation}
where
\begin{equation}\label{6.20}
\widetilde{\overline{\tau}}_{\alpha\beta}=l_{\alpha}n_{\beta}+n_{\alpha}l_{\beta}+
x_{\alpha}x_{\beta}+y_{\alpha}y_{\beta}
\end{equation}
    \item[$\bullet$]Now a direct calculation using (\ref{6.16}),
    (\ref{6.17}), (\ref{6.20}), (\ref{6.9}) and (\ref{6.12}) gives:
\end{enumerate}
\begin{equation}\label{6.21}
\widetilde{\overline{\tau}}_{\alpha\beta}\widetilde{V}^{\alpha}\widetilde{W}^{\beta}
=NM'+MN'+PP'+QQ.
\end{equation}
But if $\widetilde{V}=(\widetilde{V}^{\alpha})$ and
$\widetilde{W}=(\widetilde{W}^{\alpha})$ are future pointing
time-like vectors, in (\ref{6.16}) and (\ref{6.17}), we add:
\begin{equation}\label{6.22}
\widetilde{V}^{\alpha}\widetilde{V}_{\alpha}\leq0; \;\;
\widetilde{W}^{\alpha}\widetilde{W}_{\alpha}\leq0; \;\; M>0; \;\;
M'>0.
\end{equation}
Now (\ref{6.22}) writes, using (\ref{6.16}), (\ref{6.17}),
(\ref{6.9}) and  (\ref{6.12}):
\begin{equation*}
 -2MN+P^2+Q^2\leq 0;\;-2M'N'+P'^2+Q'^2\leq 0    ;\; M>0; \; M'>0;
\end{equation*}
so we have: $P^2+Q^2 \leq 2MN$; $P'^2+Q'^2\leq2M'N'$; $M>0$; $M'>0$.
But this implies since $M>0$; $M'>0$:
\begin{equation}\label{6.23}
N\geq 0;\;  N'\geq 0; \;MN\geq \frac{P^{2}+Q^{2}}{2} ;\;M'N'\geq
\frac{P'^{2}+Q'^{2}}{2}.
\end{equation}
Then, since $M>0$; $M'>0$, (\ref{6.23}) gives:
\begin{equation}\label{6.24}
NM'\geq \frac{M'}{2M}(P^{2}+Q^{2}) ;\;N'M\geq
\frac{M}{2M'}(P'^{2}+Q'^{2});
\end{equation}
So if we consider (\ref{6.21}) in which
$\widetilde{V}=(\widetilde{V}^{\alpha})$ and
$\widetilde{W}=(\widetilde{V}^{\alpha})$ are future pointing, we
deduce from (\ref{6.24}) that:
\begin{equation}\label{6.25}
\widetilde{\overline{\tau}}_{\alpha\beta}\widetilde{V}^{\alpha}\widetilde{W}^{\beta}\geq
\frac{1}{2MM'}\Big[M'^2(P^{2}+Q^{2})+M^{2}(P'^{2}+Q'^{2})+2MM'(PP'+QQ')\Big].
\end{equation}
Considering the term in the square bracket in the r.h.s of
(\ref{6.25}) as a quadratic polynomial in M, its discriminant is:
\begin{equation*}
    \Delta=M'^{2}[(PP'+QQ')^{2}-(P^{2}+Q^{2})(P'^{2}+Q'^{2})].
\end{equation*}
But by the properties of the usual scalar product in
$\mathbb{R}^{2}$:
\begin{equation*}
(PP'+QQ')^{2}\leq(P^{2}+Q^{2})(P'^{2}+Q'^{2})
\end{equation*}
then $\Delta \leq 0$ and the r.h.s of (\ref{6.25}) remains positive,
and so is the l.h.s. Then in this case, (\ref{6.8}) follows from
(\ref{6.19}).

\underline{$2^{\circ})$ Consider the form (\ref{6.11})}.\\
A direct calculation using (\ref{6.9}) and (\ref{6.12}) gives this
time:
\begin{equation*}
\widetilde{F}^{\lambda\mu}\widetilde{F}_{\lambda\mu}=0
;\;\;\widetilde{F}_{\alpha\lambda}\widetilde{F}^{\;\;\;\lambda}_{\beta}=\frac{C^{2}}{4}l_{\alpha}l_{\beta}.
\end{equation*}
Then (\ref{6.13}) gives
\begin{equation}\label{6.26}
   \widetilde{\tau}_{\alpha\beta}=\frac{C^{2}}{4}l_{\alpha}l_{\beta}
\end{equation}
so if $\widetilde{V}=(\widetilde{V}^{\alpha})$ and
$\widetilde{W}=(\widetilde{V}^{\alpha})$ are two future pointing
time-like vectors, using the decomposition (\ref{6.16}) and
(\ref{6.17}), we obtain from (\ref{6.9}) and (\ref{6.12}):
\begin{equation}\label{6.27}
\widetilde{\tau}_{\alpha\beta}\widetilde{V}^{\alpha}\widetilde{W}_{\beta}
=\frac{C^2}{4}(l_{\alpha}\widetilde{V}^{\alpha})(l_{\beta}\widetilde{W}^{\beta})
=\frac{C^2}{4}(-N)(-N')=\frac{C^2}{4}NN'.
\end{equation}
But (\ref{6.23}) which holds since $\widetilde{V}$ and
$\widetilde{W}$ are future pointing, gives $N\geq 0$ and  $N'\geq
0$. Hence, by (\ref{6.27})
$\widetilde{\tau}_{\alpha\beta}\widetilde{V}^{\alpha}\widetilde{W}^{\beta}
\geq0$. This completes the proof of Proposition~\ref{prop6.2}.
\end{proof}
\begin{theorem}\label{ttheo6.3}
\hspace*{2cm}\\ The global solution of the coupled
Einstein-Maxwell-Scalar Field satisfies:\begin{enumerate}
    \item[$1^{\circ})$]the weak and the dominant energy conditions.
    \item[$2^{\circ})$]the strong energy condition if
    $ \Lambda \geq \frac{(H(0))^{2}}{2}.$
\end{enumerate}
\end{theorem}
\begin{proof}
\hspace*{2cm}\\
\begin{enumerate}
    \item[$1^{\circ})$]we first prove that the solution satisfies
    the dominant energy condition (\ref{6.3}). Let $\widetilde{V}=(\widetilde{V}^{\alpha})$ and
$\widetilde{W}=(\widetilde{W}^{\alpha})$ be two future pointing
time-like vectors. By (\ref{6.8}) we have
\begin{equation}\label{6.28}
\widetilde{\tau}_{\alpha\beta}\widetilde{V}^{\alpha}\widetilde{W}^{\beta}\geq
0.
\end{equation}
Next we have, since $\widetilde{u}=(\widetilde{u}^{\alpha})$ is a
time-like future pointing vector and using (\ref{6.4}):
$$(\rho\widetilde{u}_{\alpha}\widetilde{u}_{\beta})(\widetilde{V}^{\alpha}\widetilde{W}^{\beta})=
\rho(\widetilde{u}_{\alpha}\widetilde{V}^{\alpha})(\widetilde{u}_{\beta}\widetilde{W}^{\beta})\geq
0 ;$$ so:
\begin{equation}\label{6.29}
(\rho\widetilde{u}_{\alpha}\widetilde{u}_{\beta})(\widetilde{V}^{\alpha}\widetilde{W}^{\beta})\geq
0.
\end{equation}
Now the expression (\ref{1.23}) of $\widetilde{T}_{\alpha\beta}$
gives:
\begin{eqnarray*}
\widetilde{T}_{\alpha\beta}\widetilde{V}^{\alpha}\widetilde{W}^{\beta}
   &=&\widetilde{T}_{00}\widetilde{V}^{0}\widetilde{W}^{0}+
   \widetilde{T}_{ij}\widetilde{V}^{i}\widetilde{W}^{j}  \\
   &=&\frac{1}{2}(\dot{\Phi}^2+m^2\Phi^2)\widetilde{V}^{0}\widetilde{W}^{0}+
\frac{1}{2}(\dot{\Phi}^2-m^2\Phi^2)g_{ij}\widetilde{V}^{i}\widetilde{W}^{j}
\end{eqnarray*}
then:
\begin{equation}\label{6.30}
\widetilde{T}_{\alpha\beta}\widetilde{V}^{\alpha}\widetilde{W}^{\beta}
=\frac{1}{2}\dot{\Phi}^2(\widetilde{V}^{0}\widetilde{W}^{0}+g_{ij}\widetilde{V}^{i}\widetilde{W}^{j})
+\frac{1}{2}m^2\Phi^2(\widetilde{V}^{0}\widetilde{W}^{0}-g_{ij}\widetilde{V}^{i}\widetilde{W}^{j}).
\end{equation}
But by (\ref{6.4}) which is equivalent to (\ref{6.5}), the last term
in the r.h.s of (\ref{6.30}) is positive.\\
Next, since $g$ is a scalar product and $\widetilde{V}$,
$\widetilde{W}$ are future pointing vectors, we deduce from
(\ref{6.6}) and (\ref{6.7}), that
$$|g_{ij}\widetilde{V}^{i}\widetilde{W}^{j}|\leq\widetilde{V}^{0}\widetilde{W}^{0}
;$$ then
\begin{equation}\label{6.31}
g_{ij}\widetilde{V}^{i}\widetilde{W}^{j}\geq
-\widetilde{V}^{0}\widetilde{W}^{0}.
\end{equation}
(\ref{6.31}) then shows that the first term in the r.h.s of
(\ref{6.30}) is also positive. Consequently,
\begin{equation}\label{6.32}
\widetilde{T}_{\alpha\beta}\widetilde{V}^{\alpha}\widetilde{W}^{\beta}\geq
0
\end{equation}
(\ref{6.3}) then follows from (\ref{6.28}), (\ref{6.29}) and
(\ref{6.32}). Hence the dominant energy condition (\ref{6.3}) is
satisfied.
\begin{enumerate}
    \item [$\bullet$]Setting in (\ref{6.3}),
    $\widetilde{W}=\widetilde{V}$, we have (\ref{6.31}). Hence the
    weak energy condition (\ref{6.1}) is satisfied.
\end{enumerate}
    \item[$2^{\circ})$]We now prove the strong energy condition
    (\ref{6.2}).
\begin{enumerate}
    \item [$\bullet$]Let $\widetilde{V}=(\widetilde{V}^{\alpha})$  a  future pointing
time-like vector. We deduce from the Einstein equations (\ref{1.2})
that:
\end{enumerate}
\begin{equation}\label{6.33}
\widetilde{R}_{\alpha\beta}\widetilde{V}^{\alpha}\widetilde{V}^{\beta}
=(\frac{1}{2}\widetilde{R}-\Lambda)\widetilde{V}^{\alpha}\widetilde{V}_{\alpha}
+8\pi(\widetilde{T}_{\alpha\beta}+\widetilde{\tau}_{\alpha\beta}+
\rho
\widetilde{u}_{\alpha}\widetilde{u}_{\beta})\widetilde{V}^{\alpha}\widetilde{V}^{\beta}.
\end{equation}
Since (\ref{6.1}) is satisfied, the second term in the r.h.s of
(\ref{6.33}) is positive. In the first we have :
\begin{equation}\label{6.34}
\widetilde{R}=\widetilde{g}^{\alpha\beta}\widetilde{R}_{\alpha\beta}=
\widetilde{g}^{00}\widetilde{R}_{00}+
g^{ij}\widetilde{R}_{ij}=-\widetilde{R}_{00}+g^{ij}\widetilde{R}_{ij}.
\end{equation}
Recall the classical formula linking  $\widetilde{R}_{ij}$ and
$R_{ij}$:
\begin{equation*}
\widetilde{R}_{ij}=R_{ij}-\partial_{t}k_{ij}+Hk_{ij}-2k_{il}k^{l}_{j}.
\end{equation*}
Contracting by $g^{ij}$ yields:
\begin{equation}\label{6.35}
    g^{ij}\widetilde{R}_{ij}=R-g^{ij}\partial_{t}k_{ij}+H^2-2k_{ij}k^{ij}
\end{equation}
write:
$g^{ij}\partial_{t}k_{ij}=\partial_{t}(g^{ij}k_{ij})-k_{ij}\partial_{t}g^{ij}=\partial_{t}H-2k_{ij}k^{ij}$,
(\ref{6.35}) then gives:
\begin{equation}\label{6.36}
 g^{ij}\widetilde{R}_{ij}=R-\partial_{t}H+H^{2}.
\end{equation}
Now setting in the Einstein equations $\alpha=\beta=0$ yields:
\begin{equation}\label{6.37}
  \widetilde{R}_{00}=-\frac{\widetilde{R}}{2}+\Lambda+8\pi (\widetilde{T}_{00}+\widetilde{\tau}_{00}+
\rho \widetilde{u}^{2}_{0}).
\end{equation}
(\ref{6.34}) gives, using (\ref{6.35}) using (\ref{6.37}):
\begin{equation*}
\widetilde{R}=\frac{\widetilde{R}}{2}-\Lambda-8\pi(\widetilde{T}_{00}+\widetilde{\tau}_{00}+
\rho \widetilde{u}^{2}_{0})+R-\partial_{t}H+H^{2}.
\end{equation*}
From where we deduce, using
$\widetilde{T}_{00}+\widetilde{\tau}_{00}+ \rho
\widetilde{u}^{2}_{0}\geq 0$, $R<0$, $-\partial_{t}H\leq 0$, (given
by (\ref{2.22})):
\begin{equation}\label{6.38}
\frac{\widetilde{R}}{2}-\Lambda \leq -2\Lambda+H^2.
\end{equation}
But by (\ref{2.15}) we have: $H^{2}\leq (H(0))^2$. Hence
(\ref{6.38}) gives:
\begin{equation}\label{6.39}
  \frac{\widetilde{R}}{2}-\Lambda\leq -2\Lambda+(H(0))^2
\end{equation}
but by hypothesis: $-2\Lambda +(H(0))^{2}\leq 0$; hence:
\begin{equation}\label{6.40}
\frac{\widetilde{R}}{2}-\Lambda \leq 0;
\end{equation}
since $\widetilde{V}^{\alpha}\widetilde{V}_{\alpha}< 0$,
(\ref{6.40}) implies that the first term in the r.h.s of
(\ref{6.33}) is positive; we conclude that; we have:
$\widetilde{R}_{\alpha\beta}\widetilde{V}^{\alpha}\widetilde{V}^{\beta}\geq
0$. This completes the proof of Theorem~\ref{ttheo6.3}.
\end{enumerate}
\end{proof}
\section*{Concluding Remarks}
In our future investigations, we will take into account the aspect
"distribution", of the charged particles. For this purpose we will
couple the Vlasov (resp. Boltzmann) equation in the
collisionless(resp.collisional) case.
\newpage

\section{Appendices}\label{para7}
\subsection*{$A_{1}$. \underline{Proof of formula (\ref{1.22})}}
We use the Codazzi equations which write:
\begin{equation}\label{7.1}
    \widetilde{R}_{\lambda i,
    jl}\widetilde{n}^{\lambda}=-\nabla_{l}k_{ij}+\nabla_{j}k_{il}.
\end{equation}
But $\widetilde{n}=(1, 0, 0, 0)$ so (\ref{7.1}) gives:
\begin{equation}\label{7.2}
    \widetilde{R}_{0 i,
    jl}=-\nabla_{l}k_{ij}+\nabla_{j}k_{il}.
\end{equation}
Now the curvature tensor, on $(M, \widetilde{g})$ writes, see
\cite{2}, p.240:
\begin{equation}\label{7.3}
\widetilde{R}^{\lambda}_{\;\;\alpha, \beta\mu}=
\widetilde{e}_{\beta}(\gamma_{\alpha\mu}^{\lambda})-\widetilde{e}_{\mu}(\gamma_{\beta\alpha}^{\lambda})
+(\gamma_{\beta\nu}^{l}\gamma_{\mu\alpha}^{\nu})-(\gamma_{\mu\nu}^{\lambda}\gamma_{\beta\alpha}^{\nu})
-\widetilde{C}^{\nu}_{\beta\mu}\gamma_{\nu\alpha}^{\lambda}.
\end{equation}
In particular, taking in (\ref{7.3}): $\lambda=l$; $\alpha=j$;
$\beta=i$; $\mu=0$, we obtain, using (\ref{1.21}), (\ref{1.12}),
(\ref{1.13}), (\ref{1.18}) and (\ref{1.19}):
\begin{equation}\label{7.4}
\widetilde{R}^{l}_{ \;j,
    i0}=-\frac{d}{dt}(\gamma^{l}_{ij})-\nabla_{i}k_{j}^{l}.
\end{equation}
Now (\ref{7.2}) gives, using the symmetry properties of
$\widetilde{R}^{\lambda}_{\;\alpha, \beta\mu}$:
\begin{equation}\label{7.5}
\widetilde{R}^{l}_{ \;j, i0}=-\nabla^{l}k_{ij}+\nabla_{j}k_{i}^{l}.
\end{equation}
Equalize the two values of $\widetilde{R}^{l}_{ \;j, i0}$ provided
by (\ref{7.4}), (\ref{7.5}) to obtain (\ref{1.21}).
\subsection*{$A_{2}$. \underline{Proof of Lemma~\ref{llem1}}}
\begin{enumerate}
    \item[$1^{\circ})$]\underline{Proof of (\ref{2.1})}.\\
First deduce from the evolution (\ref{1.46}) in $g_{ij}$; using
$g^{il}g_{jl}=\delta^{i}_{j}$ that:
\begin{equation}\label{7.6}
    \frac{d g^{ij}}{dt}=2 k^{ij}.
\end{equation}
Now since $H=g^{ij}k_{ij}$, (\ref{7.6}) gives:
\begin{equation}\label{7.7'}
    \frac{d H}{dt} =g^{ij}\frac{d k_{ij}}{dt}+2k_{ij}k^{ij}.
\end{equation}
A direct calculation using the evolution equation (\ref{1.47}) in
$k_{ij}$ and the relation
$\widetilde{g}^{\alpha\beta}\widetilde{\tau}_{\alpha\beta}=0$ which
implies $g^{ij}\tau_{ij}=\tau_{00}$, gives
\begin{equation}\label{7.7}
g^{ij}\frac{d k_{ij}}{dt}=R+H^2-2k_{ij}k^{ij}+4\pi
g^{ij}(T_{ij}+\rho u_{i}u_{j})-12\pi(T_{00}+\rho u_0^2)-8\pi
\tau_{00}-3\Lambda.
\end{equation}
then (\ref{2.1}) follows from (\ref{7.6}) and (\ref{7.7})
    \item[$2^{\circ})$]\underline{Proof of (\ref{2.2})}\\
We have $\frac{d H^{2}}{dt}=2H \frac{dH}{dt}$; then use (\ref{1.22})
to obtain (\ref{2.2}).

\item[$3^{\circ})$]\underline{Proof of (\ref{2.3})}\\
We have
\begin{equation}\label{7.8}
    \frac{d(k_{ij}k^{ij})}{dt}=k_{ij}\frac{d k^{ij}}{dt}+k^{ij}\frac{d
    k_{ij}}{dt}.
\end{equation}
We also have: $k^{ij}=g^{il}g^{jm}k_{lm}$; so, we have, using
(\ref{7.6})
\begin{equation}\label{7.9}
    \frac{d k^{ij}}{dt}=4k^{il}k_{l}^{j}+g^{il}g^{jm}\frac{d k_{lm}}{dt}
\end{equation}
(\ref{7.9}) gives by a direct calculation, using the evolution
equation (\ref{1.47}) for $k_{lm}$:
\begin{eqnarray}
 \nonumber \frac{d k^{ij}}{dt} &=& R^{ij}+Hk^{ij}+2k_{l}^{j}k^{il}-8\pi(T^{ij}+
\tau^{ij}+\rho  u^{i} u^{j}) \\
   &+&4\pi g^{ij}\Big[-T_{00}-\rho u^{2}_{0}+g^{lm}(T_{lm}+\rho u_{l}u_{m})\Big]
   -\Lambda g^{ij}.    \label{7.10}
\end{eqnarray}
(\ref{2.3}) then follows from (\ref{7.8}), (\ref{7.10}), using once
more the evolution equation (\ref{1.47}) for $k_{ij}$ to express the
last term in (\ref{7.8}).

    \item[$4^{\circ})$]\underline{Proof of (\ref{2.4}) and (\ref{2.5})}\\
We use the conservation laws:
\begin{equation}\label{7.11}
    \widetilde{\nabla}_{\alpha}(\widetilde{T}^{\alpha\beta}+\widetilde{\tau}^{\alpha\beta}+\rho
\widetilde{u}^{\alpha}\widetilde{u}^{\beta} )=0
\end{equation}
(\ref{7.11}) writes, using the formulae (\ref{1.21})
\begin{equation}\label{7.12}
    \widetilde{e}_{\alpha}(\widetilde{T}^{\alpha\beta}+\widetilde{\tau}^{\alpha\beta}+\rho
\widetilde{u}^{\alpha}\widetilde{u}^{\beta})+\widetilde{\gamma}_{\alpha\lambda}^{\alpha}
(\widetilde{T}^{\lambda\beta}+\widetilde{\tau}^{\lambda\beta}+\rho
\widetilde{u}^{\lambda}\widetilde{u}^{\beta})
+\widetilde{\gamma}_{\alpha \lambda}^{\beta}(\widetilde{T}^{\alpha
\lambda}+\widetilde{\tau}^{\alpha \lambda}+\rho
\widetilde{u}^{\alpha}\widetilde{u}^{\lambda})=0.
\end{equation}
It shows useful to write (\ref{7.10}) in the form:
\begin{equation*}
\widetilde{e}_{0}(\widetilde{T}^{0\beta}+\widetilde{\tau}^{0\beta}+\rho
\widetilde{u}^{0}\widetilde{u}^{\beta})+e_{l}
(\widetilde{T}^{l\beta}+\widetilde{\tau}^{l\beta}+\rho
u^{l}\widetilde{u}^{\beta}) +\widetilde{\gamma}_{\alpha
0}^{\alpha}(\widetilde{T}^{0 \beta}+\widetilde{\tau}^{0\beta}+\rho
\widetilde{u}^{0}\widetilde{u}^{\beta})
\end{equation*}
\begin{equation}\label{7.13}
+\widetilde{\gamma}_{\alpha
i}^{\alpha}(\widetilde{T}^{i\beta}+\widetilde{\tau}^{i\beta}+\rho
\widetilde{u}^{i}\widetilde{u}^{\beta})+\widetilde{\gamma}_{\alpha
0}^{\beta} (\widetilde{T}^{\alpha 0}+\widetilde{\tau}^{\alpha
0}+\rho \widetilde{u}^{\alpha}\widetilde{u}^{0})
+\widetilde{\gamma}_{\alpha j}^{\beta}(\widetilde{T}^{\alpha
j}+\widetilde{\tau}^{\alpha j}+\rho
\widetilde{u}^{\alpha}\widetilde{u}^{j})=0.
\end{equation}
Then: set in (\ref{7.13}) $\beta=0$ and use (\ref{1.19}) to obtain
(\ref{2.4}).\\
\hspace*{-2cm}set in (\ref{7.13}) $\beta=j$ and use (\ref{1.19}) to
obtain (\ref{2.5}).

    \item[$5^{\circ})$]\underline{Proof of (\ref{2.6})}\\
\end{enumerate}
We have $R=g^{ij}R_{ij}$, where $R_{ij}$ given by (\ref{1.55}),
then: $$\frac{d R}{dt}=R_{ij}\frac{d g^{ij}}{dt}+g^{ij}\frac{d
R_{ij}}{dt}.
$$ So we have, using (\ref{7.6})
\begin{equation}\label{7.14}
\frac{d R}{dt}=2 R_{ij}k^{ij}+g^{ij}\frac{d R_{ij}}{dt}.
\end{equation}
Now we verifies by a direct calculation, using the formula
(\ref{1.21}) and the expression (\ref{1.55}) of $R_{ij}$ that:
\begin{equation}\label{7.15}
    \frac{d R_{ij}}{dt}=\nabla_{l}\Big(\frac{d \gamma^{l}_{ij}}{dt}\Big)
\end{equation}
(\ref{7.15})gives, using formula (\ref{1.22}) and since
$H=g^{ij}k_{ij}$ depends only on $t$:
\begin{equation}\label{7.16}
    g^{ij}\frac{d R_{ij}}{dt}=-2\nabla_{l}\nabla_{i}k^{il}
\end{equation}
(\ref{2.6}) then follows from (\ref{7.14}) and (\ref{7.16}).
\subsection*{$A_{3}$. \underline{Proof of (\ref{2.8})}}
It is easily seen, using the definition of $A_{j}$ in
Lemma~\ref{llem2}, that
\begin{equation}\label{7.17}
    A_{l}=g_{jl}B^{j}
\end{equation}
where:
\begin{equation}\label{7.18}
    B^{j}=\nabla_{i}k^{ij}-8\pi(T^{0j}+\tau^{0j}+\rho
u^{0}u^{j}).
\end{equation}
We then have, differentiating (\ref{7.17}) and using equation
(\ref{1.46}) in $g_{ij}$:
\begin{equation}\label{7.19}
    \frac{d A_l}{dt} = -2k_{jl}B^j+g_{jl}\frac{d B^j}{dt}.
\end{equation}
Now we have by (\ref{7.18}):
\begin{equation}\label{7.20}
 \frac{d B^{j}}{dt}=\frac{d}{dt}(\nabla_{i}k^{ij})-8\pi\frac{d}{dt}(T^{0j}+\tau^{0j}+\rho u^{0}u^{j})
\end{equation}
But
$\nabla_{i}k^{ij}=\partial_{i}k^{ij}+\gamma^{i}_{il}k^{lj}+\gamma_{il}^{j}k^{il}$;
then we have:
\begin{equation}\label{7.21}
 \frac{d(\nabla_{i}k^{ij})}{dt}=\partial_{i}(\frac{d k^{ij}}{dt})+ \gamma_{il}^{i}\frac{dk^{lj}}{dt}
 +\gamma_{il}^{j}(\frac{d k^{il}}{dt})
   + (\frac{d \gamma_{il}^{i}}{dt}) k^{lj}+(\frac{d
   \gamma^{j}_{il}}{dt})k^{il}.
\end{equation}
Now use (\ref{7.10}) to express $\frac{d k^{ij}}{dt}$, (\ref{1.22})
to express $\frac{d \gamma_{il}^{j}}{dt}$ and obtain:
\begin{eqnarray}
\nonumber\frac{d(\nabla_{i}k^{ij})}{dt}
&=&\nabla_{i}R^{ij}+H\nabla_{i}k^{ij}+2(\nabla_{i}k^{il})k_{l}^{j}+2k^{il}\nabla_{i}k_{l}^{j}-8\pi\nabla_{i}(T^{ij}+\tau^{ij}+\rho  u^{i}u^{j})  \\
            &+&(\nabla^{j}k_{il}-\nabla_{i}k^{j}_{l}+\nabla_{l}k_{i}^{j})k^{il}   \label{7.22}
\end{eqnarray}
(\ref{7.20}) then gives using (\ref{7.22}) to express the first
term, (\ref{2.5}) to express the second term, and taking into
account the expression (\ref{7.18}) of $B^{j}$.
\begin{equation}\label{7.23}
    \frac{dB^{j}}{dt}=HB^{j}+2k_{i}^{j}B^{i}+\nabla_{i}R^{ij}
\end{equation}
But by the Bianchi identities:
\begin{equation*}
    \nabla_{i}R^{ij}=\frac{1}{2}\nabla^{i}R=g^{ij}\nabla_{j} R=0
\end{equation*}
since R depends only on $t$.\\
Finally (\ref{2.8}) follows from (\ref{7.19}), (\ref{7.23}) and
(\ref{7.17})


\begin{thebibliography}{122}
\bibitem{1}
Alcubierre M. 2008 \emph{Introduction to 3+1 numerical relativity}\\
Oxford: Oxford University Press.

\bibitem{2}
Yvonne Choquet-Bruhat, \emph{G\'eom\'etrie Diff\'erentielle et
Syst\`emes Ext\'erieurs}, Dunod Paris, 1968.

\bibitem{3}
 Choquet-Bruhat Y. De Witt-Morette, C and  Dillard-Bleick, M.
 1997,\\
 \emph{Analysis, Manifolds and Physics 1}, (Amsterdam,
 North-Holland).

\bibitem{4} D. Christodoulou, \emph{Bounded variation solution of the
spherically symmetric Einstein-scalar-fields equations},\\
Comm.Pure.Appl.Math 46 (1993) 1131-1220.

\bibitem{5} L. Derone, \emph{Le syst\`eme de d\'etection de
l'exp\'erience VIRGO d\'edi\'ee \'a la recherche d'ondes
gravitationnelles}.\\ Th\`ese
(1999):http://fr.wikipedia.org/wiki/portail.ondes gravitationnelles.

\bibitem{6} E. Gourgoulhon. \emph{3 + 1 Formalism and bases of numerical Relativity}.
Preprint:http://arxiv.org/abs/gr-qc/0703035v1 (2007).

\bibitem{7}
Hayoung Lee 2004. \emph{Asymptotic behaviour of the Einstein -vlasov
system with a positive  cosmological constant},
 Math. Proc. Comb. Phil. Soc.137, 495-509.

\bibitem{8}
Hayoung Lee. \emph{The Einstein-Vlasov system with a scalar field}:
Ann. H. Poincar\'e 6, 687-723 (2005).

\bibitem{9}Hawking SW and Ellis FR, 1973, \emph{The large scale
structure of space-time} (Cambridge Monographs and Maths. Phys)
Cambridge: Cambridge University Press.

\bibitem{10} Jantzen RT 1984 \emph{Cosmology of the early universe}\\
ed LZ Fang and R Ruffini (Singapore: world scientific).

\bibitem{11}
Kitada, Y. and Maeda, K. \emph{Cosmic no-hair theorem in homogeneous
spacetimes I Bianchi models}. Class. Quantum Grav. 10, 703-734
(1993).

\bibitem{12}
 Lichnerowicz, A: \emph{th\'eories relativistes de la gravitation et de l'\'electromagn\'etisme}.
Masson et Cie Edition, (1995).

\bibitem{13}
Moss, I.and Sahni,V.,  \emph{Anisotropy in the chaotic inflationary
universe}. Phys. Lett. B178, 159-162 (1983).
\bibitem{14}
N.Noutchegueme and E. Takou, \emph{Global existence of solutions for
the Einstein-Boltzmann system  with cosmological constant in a
Friedman-Robertson-Walker space-time},  Comm. Math. Sci 4(2) (2006)
295-314.

\bibitem{15}
N.Noutchegueme and G. Chendjou, \emph{Global solutions to the
Einstein equations with cosmological constant on
Friedman-Robertson-Walker space-time with plane, hyperbolic and
spherical symmetries}, Comm. Math. Sci 6(3) (2008) 595-610.

\bibitem{16}
N. Noutchegueme and E. M. Tetsadjio. \emph{Global dynamics for a
collisionless charged plasma in Bianchi spacetimes}. Class. Quantum
Grav 26(2009) 195001 (16pp).

\bibitem{17}
 A.D. Rendall, \emph{Cosmic censorship for some spatially homogeneous cosmological
    models}. Ann. Phys. 233 82-96 (1994).

\bibitem{18}
A.D. Rendall: \emph{on the nature of singularities in plane symmetry
scalar field cosmologies}, Gen. Relativity and gravitation 27 (1995)
213-221.

\bibitem{19}
A.D. Rendall. \emph{Global properties of locally spatially
homogeneous cosmological models with matter}. Math. Proc. Camb.
Phil. Sco. \textbf{118} (1995), 511-526.

\bibitem{20}
A.D. Rendall: \emph{Accelerated cosmological expansion due to a
scalar field whose potential has a positive lower bound}. Class.
Quantum Grav. 21, 2445-2454 (2004).

\bibitem{21}
A.D. Rendall. \emph{Partial Differential Equation in General
Relativity}, Oxford Graduate text in Mathematics, Vol 16(2008).

\bibitem{22}
N.Straumann, \emph{On the cosmological constant problems and the
astronomical evidence for the homogeneous energy density with
negative pressure in, Vacuum Energy}, Renormalisation eds. B.
Duplantier and V. Rivasseau. (Birkhausser, Basel, (2003).

\bibitem{23}
S.B.Tchapnda and N.Noutchegueme: \emph{the surface symmetric
Einstein-Vlasov system with cosmological constant},
Math.Proc.Cambridge Phil.Soc 138, (2005)541-724.

\bibitem{24}
D.Tegankong, N.Noutchegueme and  A.D.Rendall: \emph{Local existence
and continuation criteria for solutions of the
Einstein-Vlasov-Scalar field system}. J.Hyperbolic Differential
Equations 1(4) (2004)691-724.

\bibitem{25}
 Wainwright J. and  Ellis, FR, 1997, \emph{Dynamical systems in cosmology},
(Cambridge: Cambridge University Press).

\bibitem{26}
Wald, R, 1983, \emph{Asymptotic behaviour of homogeneous
cosmological models in the presence of a positive cosmological
constant}. Phys.Review. D 28, 2118-2120.

\bibitem{27}
Wald, R. 1984 \emph{General Relativity} (Chicago II: University of
Chicago Press).
\end{thebibliography}
\end{document}